\documentclass[journal,10pt,twoside]{IEEEtran}

\usepackage{amsmath,amssymb,amsfonts,bm,mathtools}
\usepackage{graphicx}
\usepackage{booktabs}
\usepackage{array}
\usepackage{multirow}
\usepackage{xcolor}
\usepackage{tikz}
\usetikzlibrary{arrows.meta,positioning,calc,3d}
\usepackage{algorithm}
\usepackage{algpseudocode}
\usepackage[colorlinks=true,allcolors=blue]{hyperref}
\newtheorem{theorem}{Theorem}
\newtheorem{lemma}{Lemma}
\newtheorem{definition}{Definition}
\usepackage{cite}

\DeclareMathOperator{\Trace}{Tr}
\DeclareMathOperator{\Var}{Var}
\newcommand{\Jm}{\hat{J}^{-}}
\newcommand{\Jp}{\hat{J}^{+}}
\newcommand{\Aop}{\hat{A}}
\newcommand{\Gone}{G^{(1)}}
\newcommand{\Gtwo}{G^{(2)}}
\newcommand{\eNR}{\eta_{\mathrm{NR}}}
\newcommand{\FQ}{\mathcal{F}_{Q}}
\newcommand{\FQM}{\bm{\mathcal{F}}_{Q}}
\newcommand{\ketv}[1]{\lvert #1 \rangle}
\newcommand{\brav}[1]{\langle #1 \rvert}

\begin{document}


\title{Theoretical Investigation of Nonclassical Directional Emission in 1D and
2D Quantum-Dot Antenna Arrays and Quantum DoA Estimation}


\author{Angeru~Pavan~Kumar,~\IEEEmembership{Student~Member,~IEEE,}
and~Sandip~Ghosal,~\IEEEmembership{Member,~IEEE}
\thanks{\color{red}This work has been submitted to the IEEE for possible publication.
Copyright may be transferred without notice, after which this version may no
longer be accessible.}
\thanks{A. Pavan Kumar (E-mail: 225ec8001@nitrkl.ac.in) and S. Ghosal (Corresponding Author, E-mail: sgmw204@gmail.com) are with the Department of Electronics and Communication
Engineering, National Institute of Technology Rourkela, Rourkela 769008,
India.}}

\markboth{IEEE TRANSACTIONS ON ANTENNAS AND PROPAGATION}%
{Kumar and Ghosal: Investigation of Nonclassical Directional Emission in QD Antenna Arrays}

\maketitle

\begin{abstract}
A quantum-dot array prepared in a timed Dicke state exhibits directional radiation asymmetry, which has often been interpreted as evidence of reciprocity breaking through quantum-state preparation. We show that this interpretation is incorrect: reciprocity concerns transmit-receive response, not necessarily pattern symmetry. Moreover, the first-order coherence of any single-excitation state is exactly equivalent to the power pattern of a classical phased array driven by conjugate weights, making the observed asymmetry classically reproducible. The genuine quantum signature instead appears in second-order correlations, where a single-excitation state violates the classical lower bound on photon coincidences. We therefore propose a quantum-antenna certificate combining nonzero directional asymmetry with photon antibunching at the beam peak. We further derive a tight asymmetry bound, establish the correct mirror operation for planar arrays, correct the collective decay rate with dipole-dipole coupling, and derive a saturation law for direction-finding precision. All theoretical results are validated through independent numerical implementations.

\end{abstract}


\section{Introduction}

\IEEEPARstart{T}{he} far-field radiation patterns of the classical antenna arrays depend upon the array factor, which is uniquely determined by the beamforming weight vector applied via an external feed network \cite{Balanis}.
Two fundamental laws govern every classical passive
antenna, regardless of its geometry, size, or feeding
configuration.
\begin{itemize}
    \item Firstly, the Lorentz reciprocity theorem -- a mathematical consequence of time-reversal symmetry in Maxwell's equations -- enforces that the radiation pattern is reciprocally symmetric for all observation angles.
    \item Secondly, the radiation pattern is fixed entirely by the
external weight vector; changing the pattern always requires a physical hardware change. Breaking reciprocity classically requires ferrite materials, active amplifiers, or time-varying circuits --all external hardware interventions.
\end{itemize}

With advances in TeraHertz and quantum technologies, quantum dot (QD) antenna arrays have gained significant research interest in recent years \cite{chewPart1} and \cite{chewPart2}. However, the fundamental operating principles of QD arrays differ from those of classical antenna arrays in several respects \cite{9096523} and \cite{9286841}. It is essential to separate the Lorentz reciprocity from \emph{directional pattern
asymmetry}, $G(\hat{\mathbf{r}})\neq G(-\hat{\mathbf{r}})$, which is what the
parameter $\eNR$ used throughout the QD antenna literature actually measures.
Reciprocity does \emph{not} require $G(\hat{\mathbf{r}})=G(-\hat{\mathbf{r}})$:
any end-fire phased array violates the latter while fully obeying the former.
We retain the historical symbol $\eNR$ for continuity with the prior QD-array
literature, but interpret it throughout strictly as a \emph{directional pattern
asymmetry} metric and never as a Lorentz-reciprocity measure. Accordingly, we
use ``directional pattern asymmetry'' for the quantity itself and reserve
``nonreciprocity'' for describing how earlier work framed the effect.
Conflating the two has led to pattern asymmetry being reported as a quantum
effect. Section~\ref{sec:lemma} shows it is classically reproducible, and the
remainder of the paper supplies a criterion that is not.\\
Quantum dot arrays (QDAs) offer a fundamentally different paradigm, one where the quantum state of the system replaces the classical weight vector. At the nanoscale, QDAs emit through discrete quantum transitions described by quantum electrodynamics rather than classical Maxwellian theory. Quantum dot arrays (QDAs) replace $\{a_n\}$ with the density matrix $\rho$.
Slepyan and Boag reported directional asymmetry for an array in a timed Dicke
state and interpreted it as quantum nonreciprocity~\cite{SlepyanBoag}. Here we
distinguish that directional asymmetry from Lorentz reciprocity and show that
the first-order pattern is classically reproducible.
Komarov and Slepyan treated QDAs as open systems~\cite{Komarov}, Mikhalychev
\textit{et al.} synthesised field correlations~\cite{Mikhalychev}, and reviews
appear in~\cite{SlepyanRev,SlepyanRadar,Abohmra}. Collective decay follows
Dicke~\cite{Dicke}, with universality in ordered arrays established
in~\cite{Masson} and dimensionality effects in~\cite{Sierra}.\\
In parallel, quantum sensing has advanced DoA estimation. Richardson
\textit{et al.} studied AoA with simulated Rydberg linear arrays~\cite{Richardson},
Yan \textit{et al.} demonstrated an L-shaped Rydberg locator~\cite{Yan}, and
Xu \textit{et al.} reported a Rydberg superheterodyne $3\times3$ receiver for
2-D DoA~\cite{Xu}. Kim \textit{et al.} introduced Quantum
MUSIC~\cite{Kim}, extended to Rydberg arrays with magnitude-only observables
by Banerjee \textit{et al.}~\cite{Banerjee}. Precision limits follow from the
quantum Fisher information (QFI) and the quantum Cram\'er--Rao bound
(QCRB)~\cite{Helstrom,Liu}.\\
Following the previous literature, three major research gaps are identified as:
\begin{itemize}
    \item No QD quantum antenna framework has been extended to a planar geometry to date.
    \item There exists no closed-form quantum spatial coherence matrix for
2-D QD arrays yet.
\item To date, no QFI-based precision bound has been derived for QD antenna
states.
\end{itemize}

This paper addresses all three open research gaps within an inclusive, numerically validated theoretical framework.

\subsection{Roadmap:} The argument runs in five steps.
\emph{(i)}~The prior literature reads $\eNR>0$ as evidence of quantum
nonreciprocity. \emph{(ii)}~We show this does not follow, because
$\Gone_{\mathrm{TD}}=|AF|^2$ for a classical phased array
(Sec.~\ref{sec:lemma}), and because pattern asymmetry is not a
Lorentz-reciprocity test at all. \emph{(iii)}~We supply what does follow. The paper
turns on a two-tier distinction: \emph{classical radiation-pattern equivalence}
(first order, classically reproducible) versus \emph{nonclassical photon
statistics} (second order, not reproducible). The conjunction $\eNR>0$
\emph{and} $g^{(2)}(\hat{\mathbf{r}}_0)<1$ certifies nonclassical directional
emission (Sec.~\ref{sec:cert}), with three reference states separating the two
tiers cleanly. \emph{(iv)}~We extend from 1-D to 2-D, correct the
spatial-reversal mirror, prove the bound $\eNR\leq2$, and replace $N^{2}\Gamma$ by the
state projection $\mathbf{c}^{\dagger}\bm{\Gamma}\mathbf{c}$ under
dipole--dipole coupling. \emph{(v)}~We apply the framework to direction finding
and bound the attainable precision. Readers interested only in the positive
result may go directly to Sec.~\ref{sec:cert}.

\subsection{Contributions:} This paper is partly corrective of the existing concepts and partly
contributive towards new ideas, as summarized below.
\begin{enumerate}
\item \emph{A classical-equivalence lemma} (Sec.~\ref{sec:lemma}). The timed
Dicke $\Gone$ equals the classical phased-array $|AF|^2$ to $4\times10^{-16}$.
Hence $\eNR$ alone cannot certify quantum behaviour.
\item \emph{A tight bound} $\eNR\in[0,2]$ (Theorem~\ref{thm:bound}), reached to
$89.5\%$ in 1-D and $93.3\%$ in 2-D.
\item \emph{The Nonclassical Directional-Emission Certificate} (Sec.~\ref{sec:cert}):
$\eNR>0$ together with $g^{(2)}(0)<1$. The timed Dicke state achieves
$\eNR=1.791$ with $g^{(2)}(0)=0$, which no classical field can do.
\item \emph{Corrected collective decay} (Sec.~\ref{sec:decay}) including
dipole--dipole coupling, replacing $N^{2}\Gamma$ by $\lambda_{\max}(\Gamma_{ij})$
and exposing subradiant dark states at $0.001\Gamma$.
\item \emph{The planar spatial-reversal mirror} $\phi\rightarrow\phi+\pi$
(Sec.~\ref{sec:eta2d}), without which $\eNR\equiv0$ identically.
\item \emph{A saturation law for quantum DoA}, $6(M-1)/(M+1)$
(Sec.~\ref{sec:qfi}), bounding the QFI enhancement.
\end{enumerate}

\subsection{Organization of the Paper}
The remainder of this paper is organized as follows. Section~II describes
the array geometry and the physical model of the InGaAs quantum-dot
emitter, and fixes the parameters used throughout. Section~III develops
the quantum formalism: the array factor operator, the element-basis
correlation matrix, the first- and second-order coherences, and the sign
convention on which every later result depends. Section~IV establishes
the classical-equivalence lemma, proves the tight bound on the asymmetry
metric, and introduces the proposed certificate for nonclassical
directional emission together with its measurement protocol. Section~V
revisits the collective-emission model with dipole--dipole coupling
retained and corrects the decay rate of the prepared state. Section~VI
applies the framework to direction-of-arrival estimation, deriving the
quantum spatial coherence matrix, the source model, and the saturation
law that bounds the attainable precision. Section~VII reports the
numerical results for the linear and planar arrays. Section~VIII examines
robustness against dephasing, fabrication disorder, and the vectorial
element pattern, and states the limitations of the present model.
Section~IX provides the concluding summary of the contributions.

\color{black}
\section{Array Geometry and Physical Model}
\subsection{Array Configurations}
Fig.~\ref{fig:geom} defines both geometries treated in this work. The 1-D ULA
places $N=4$ QDs along $\hat{x}$ at $x_n=nd$, $n=0,\dots,3$. The 2-D UPA places
$M\times N=4\times4=16$ QDs on the $xy$-plane at
$\mathbf{r}_{mn}=(md_x,\,nd_y,\,0)$.

\begin{figure}[!t]
\centering
\includegraphics[width=\columnwidth]{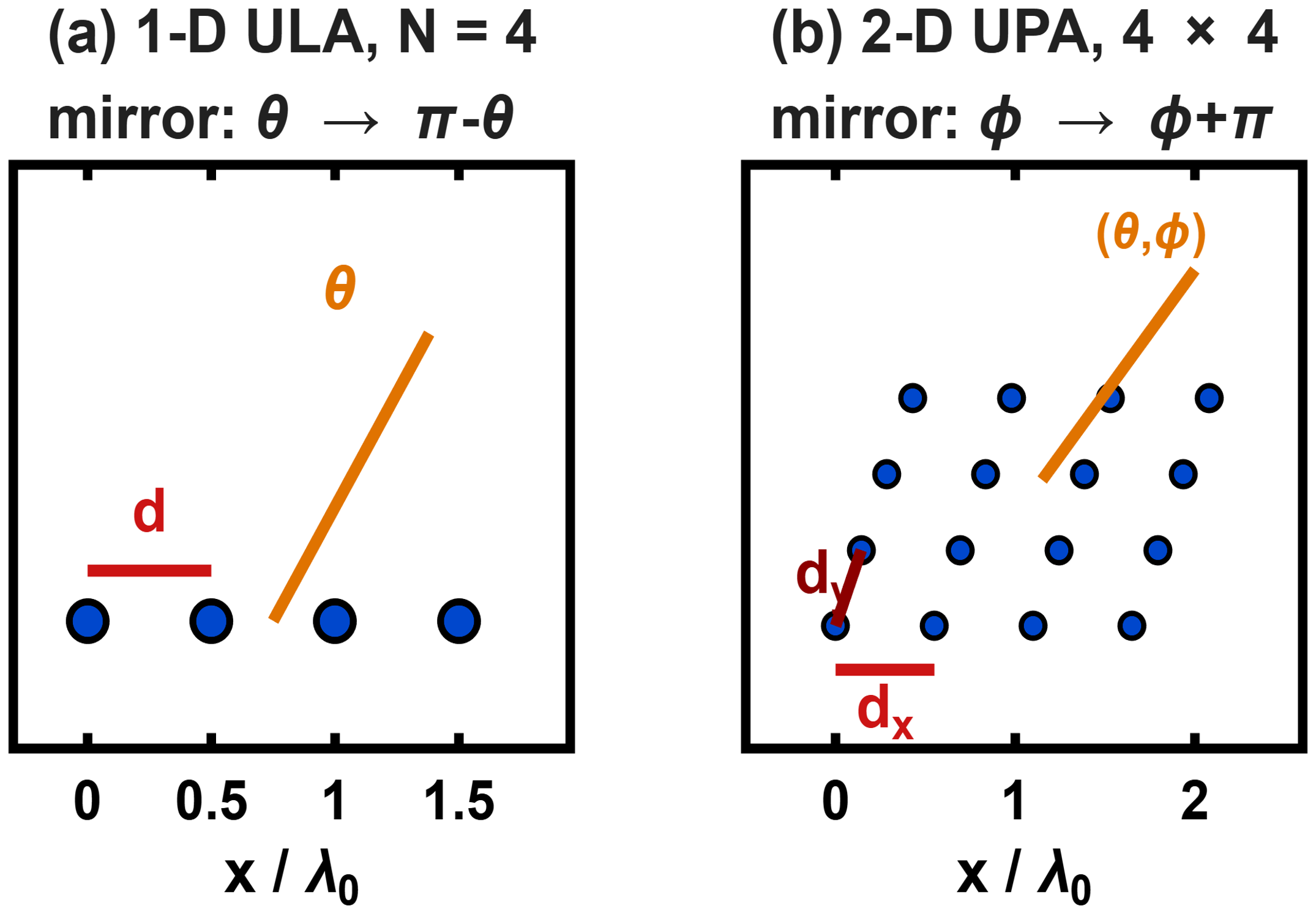}
\caption{Array geometries. (a) 1-D ULA of $N=4$ InGaAs QDs along $\hat{x}$ with
spacing $d$; the spatial-reversal mirror is $\theta\rightarrow\pi-\theta$, which flips
the direction cosine $\cos\theta$. (b) 2-D UPA of $4\times4$ QDs in the $xy$-plane.
For a planar aperture $\sin(\pi-\theta)=\sin\theta$ leaves both direction cosines
$(u,v)$ unchanged, so the correct spatial-reversal mirror is $\phi\rightarrow\phi+\pi$.
Using $\theta\rightarrow\pi-\theta$ in 2-D returns $\eNR\equiv0$ artificially.}
\label{fig:geom}
\end{figure}

\subsection{InGaAs Quantum Dot as a Two-Level Emitter}
A QD of confinement size $L$ and effective mass $m^*$ has particle-in-a-box
levels \cite{chewPart1}
\begin{equation}
E_n=\frac{n^2\pi^2\hbar^2}{2m^*L^2},\qquad n=1,2,3,\dots
\label{eq:pib}
\end{equation}
With $L=15$~nm and $m^*=0.04\,m_e$, \eqref{eq:pib} gives $E_1=0.042$~eV and
$E_2=0.167$~eV, hence
\begin{equation}
\hbar\omega_0=E_2-E_1=0.125~\mathrm{eV},\quad
\lambda_0=\frac{hc}{\hbar\omega_0}=9.89~\mu\mathrm{m}.
\end{equation}
The $n^2$ spacing in~\eqref{eq:pib} isolates the $1\!\rightarrow\!2$ transition,
justifying the two-level reduction $\hat{H}_{\mathrm{TLS}}=(\hbar\omega_0/2)\hat{\sigma}^z$.
The spontaneous emission rate follows Fermi's golden rule,
\begin{equation}
\Gamma=\frac{\omega_0^3\,|d_{ab}|^2}{3\pi\varepsilon_0\hbar c^3},~\text{and}~
\qquad
d_{ab}=\frac{16\,e\,L}{9\pi^2},
\label{eq:gamma}
\end{equation}
giving $d_{ab}=4.33\times10^{-28}$~C$\cdot$m, $\Gamma=5.46\times10^{6}$~s$^{-1}$ and
$\tau=1/\Gamma=0.183~\mu$s. Table~\ref{tab:params} collects these together with
the array parameters used throughout. Since $L=\lambda_0/659\ll\lambda_0$, each QD is modelled as a point quantum
emitter with an \emph{isotropic scalar radiation factor}. This approximation
neglects the vectorial $\sin^2\theta$ element pattern of a true electric dipole;
its implications are discussed in Sec.~\ref{sec:limits}. Collective coupling is
treated separately through the dyadic Green-function expressions of
Sec.~\ref{sec:decay}. Within this approximation the element pattern carries no
directivity, so the array pattern is set entirely by the state.\\
It is to note that~\eqref{eq:pib} is a one-dimensional infinite-well
effective-mass model adopted to obtain a tractable two-level emitter. A
fabricated InGaAs QD involves three-dimensional electron and hole confinement,
band structure, Coulomb and excitonic effects, and a specific dipole
orientation. The parameters above should therefore be read as a theoretical
effective-emitter model rather than a quantitatively complete device model.

\begin{table}[!t]
\renewcommand{\arraystretch}{1.18}
\caption{Physical and Array Parameters (All Values Verified Numerically)}
\label{tab:params}
\centering
\begin{tabular}{llcc}
\toprule
\textbf{Quantity} & \textbf{Symbol} & \textbf{1-D ULA} & \textbf{2-D UPA}\\
\midrule
QD confinement      & $L$              & \multicolumn{2}{c}{$15$~nm}\\
Effective mass      & $m^*$            & \multicolumn{2}{c}{$0.04\,m_e$}\\
Ground level        & $E_1$            & \multicolumn{2}{c}{$0.042$~eV}\\
Excited level       & $E_2$            & \multicolumn{2}{c}{$0.167$~eV}\\
Transition energy   & $\hbar\omega_0$  & \multicolumn{2}{c}{$0.125$~eV}\\
Wavelength          & $\lambda_0$      & \multicolumn{2}{c}{$9.89~\mu$m}\\
Dipole moment       & $d_{ab}$         & \multicolumn{2}{c}{$4.33\times10^{-28}$~C$\cdot$m}\\
Emission rate       & $\Gamma$         & \multicolumn{2}{c}{$5.46\times10^{6}$~s$^{-1}$}\\
Radiative lifetime  & $\tau$           & \multicolumn{2}{c}{$0.183~\mu$s}\\
Size ratio          & $L/\lambda_0$    & \multicolumn{2}{c}{$1/659$}\\
\midrule
Elements            & $N$ / $MN$       & $4$        & $16$\\
Reference spacing   & $d$              & $0.5\lambda_0$ & $0.5\lambda_0$\\
Nonrecip.\ spacing  & $d_{\mathrm{NR}}$& $0.3\lambda_0$ & $0.3\lambda_0$\\
Hilbert dimension   & $2^{N}$          & $16$       & $65\,536$\\
Dicke subspace      & $N{+}1$          & $5$        & $17$\\
Compression         & ---              & $3.2\times$& $3855\times$\\
Max collective rate & $\lambda_{\max}$ & $1.61\Gamma$ & $2.75\Gamma$\\
TD decay $\mathbf{c}^{\dagger}\bm{\Gamma}\mathbf{c}$ & $\Gamma_\psi(0)$ & $1.01\Gamma$ & $2.08\Gamma$\\
Subradiant rate     & $\lambda_{\min}$ & $0.087\Gamma$ & $0.001\Gamma$\\
Lamb shift          & $J_{ij}$         & \multicolumn{2}{c}{$0.289\Gamma$ at $0.3\lambda_0$}\\
\bottomrule
\end{tabular}
\end{table}

\section{Quantum Formalism}
\subsection{Lindblad Dynamics in the Dicke Subspace}
The array radiates irreversibly into the vacuum, so its dynamics obey the
Lindblad master equation~\cite{Lindblad,Komarov}
\begin{equation}
\dot{\rho}=-\frac{i}{\hbar}\bigl[\hat{H},\rho\bigr]
+\Gamma\!\left(\Jm\rho\Jp-\tfrac12\Jp\Jm\rho-\tfrac12\rho\Jp\Jm\right),
\label{eq:lindblad}
\end{equation}
with collective operators $\hat{J}^{\pm}=\sum_j\hat{\sigma}^{\pm}_j$. Because the
emitters are identical, the dynamics preserve permutation symmetry and remain in
the symmetric Dicke ladder $\{\ketv{J,M_J}\}$ with $J=N/2$, on which
\begin{equation}
\Jm\ketv{J,M_J}=\sqrt{(J+M_J)(J-M_J+1)}\,\ketv{J,M_J-1}.
\label{eq:dickelower}
\end{equation}
This reduces the state space from $2^{N}$ to $N{+}1$: from $65\,536$ to $17$ for
$MN=16$, a compression of $3855\times$.\\
Two distinct models are used in this paper and must not be conflated.
\emph{Model~A} is the ideal Dicke model, valid for $d\ll\lambda_0$, in which all
pairs couple identically and the dynamics remain in the symmetric
$(N{+}1)$-dimensional ladder; it is used only for the qualitative population
dynamics. \emph{Model~B} is the finite-spacing model of
Sec.~\ref{sec:decay}, in which $\Gamma_{ij}$ and $J_{ij}$ depend on separation
and the dynamics generally leave the symmetric ladder. All quantitative results
at $d=0.3\lambda_0$ use Model~B, restricted to the $N$-dimensional
single-excitation manifold where it is exactly solvable.

\subsection{Why Patterns Must Leave the Dicke Ladder}
\label{sec:manifold}
The quantum array factor operator is the substitution
$a_n\rightarrow\hat{\sigma}^-_n$ applied to the classical array factor:
\begin{equation}
\Aop(\hat{\mathbf{r}})=\sum_{j}e^{\,jk_0\,\hat{\mathbf{r}}\cdot\mathbf{r}_j}\,
\hat{\sigma}^{-}_{j}.
\label{eq:Aop}
\end{equation}
If~\eqref{eq:Aop} is projected onto the \emph{symmetric} Dicke ladder, every
$\hat{\sigma}^-_j$ maps to the same collective $\Jm$, so
\begin{equation}
\Aop\big|_{\mathrm{Dicke}}
= \underbrace{\Bigl[\textstyle\sum_j e^{\,jk_0\hat{\mathbf{r}}\cdot\mathbf{r}_j}\Bigr]}
_{\displaystyle S(\hat{\mathbf{r}})\;\in\;\mathbb{C}}\;\frac{\Jm}{N},
\end{equation}
whence
$\Gone=|S(\hat{\mathbf{r}})|^2\,\Trace[\Jp\Jm\rho]/N^2$. The quantum state then
merely \emph{scales} a purely classical array factor, and all state-dependent
directionality --- including the timed Dicke phase imprint --- is lost, forcing
$\eNR\equiv0$. The precise statement is: \emph{for the symmetric Dicke-ladder model with
identical collective coupling, the spatial phase profile of the timed Dicke
state is not represented by a state-dependent array factor.} It is not that
directional emission requires a single excitation --- general multi-excitation
states can radiate directionally --- but that this particular reduction discards
the phase. We therefore evaluate patterns in the single-excitation manifold
$\{\ketv{e_j}\}$ of dimension $N$, where
$\ketv{\psi}=\sum_j c_j\ketv{e_j}$ retains its phase profile. Writing the
steering vector
\begin{equation}
[\mathbf{a}(\hat{\mathbf{r}})]_j=e^{\,jk_0\hat{\mathbf{r}}\cdot\mathbf{r}_j},
\label{eq:steer}
\end{equation}
\emph{Conjugation and sign convention.} Expanding
$\Trace[\Aop^{\dagger}\Aop\rho]$ with~\eqref{eq:Aop} and~\eqref{eq:steer} gives
\begin{equation}
\Gone(\hat{\mathbf{r}})
=\sum_{ij}e^{-jk_0\hat{\mathbf{r}}\cdot\mathbf{r}_i}
        e^{+jk_0\hat{\mathbf{r}}\cdot\mathbf{r}_j}\,C_{ij}
=\sum_{ij}a_i^{*}\,a_j\,C_{ij},
\label{eq:G1expand}
\end{equation}
where $C_{ij}=\Trace(\hat{\sigma}^{+}_i\hat{\sigma}^{-}_j\rho)$ is the
first-order correlation matrix in the array-element basis. This fixes the
conjugation unambiguously, and the Glauber first-order
coherence~\cite{Glauber} becomes the quadratic form
\begin{equation}
\Gone(\hat{\mathbf{r}})=\Trace\!\bigl[\Aop^{\dagger}\Aop\rho\bigr]
=\mathbf{a}^{\mathsf{H}}(\hat{\mathbf{r}})\,\mathbf{C}\,
\mathbf{a}(\hat{\mathbf{r}}),
\label{eq:G1}
\end{equation}
so that, in matrix form, $\Gone=\mathbf{a}^{\mathsf{H}}\mathbf{C}\mathbf{a}$.
For a pure single-excitation state $\ketv{\psi}=\sum_j c_j\ketv{e_j}$ one has
$\hat{\sigma}^{-}_j\ketv{\psi}=c_j\ketv{G}$ and hence
$C_{ij}=c_i^{*}c_j$, i.e.\ $\mathbf{C}=\mathbf{c}^{*}\mathbf{c}^{\mathsf{T}}$.
Substituting into~\eqref{eq:G1expand} the two conjugated factors combine into a
single modulus, it will provide
\begin{equation}
\Gone(\hat{\mathbf{r}})
=\Bigl|\sum_j c_j\,e^{+jk_0\hat{\mathbf{r}}\cdot\mathbf{r}_j}\Bigr|^{2}
=\bigl|\mathbf{a}^{\mathsf{T}}(\hat{\mathbf{r}})\,\mathbf{c}\bigr|^{2}.
\label{eq:G1pure}
\end{equation}
It is worth highlighting here \eqref{eq:G1pure} considers the \emph{transpose}, not the Hermitian conjugate. the conjugation already
sits in $\mathbf{C}$, and writing
$|\mathbf{a}^{\mathsf{H}}\mathbf{c}|^{2}$ here would conjugate twice and reverse the beam direction.\\
Consequently, for pointing the beam towards a chosen direction
$\hat{\mathbf{r}}_0$, the state phases must \emph{oppose} the field phases,
\begin{equation}
c_j=\frac{1}{\sqrt{N}}\,e^{-jk_0\hat{\mathbf{r}}_0\cdot\mathbf{r}_j},
\label{eq:conv}
\end{equation}
since then $c_j a_j=e^{+jk_0(\hat{\mathbf{r}}-\hat{\mathbf{r}}_0)\cdot\mathbf{r}_j}$
and every term equals unity at $\hat{\mathbf{r}}=\hat{\mathbf{r}}_0$. The remaining contents of this work universally follow \eqref{eq:conv}. The chosen
$\hat{\mathbf{r}}_0$ must be a physical direction, $u_0^2+v_0^2\leq1$; a phase
gradient violating this corresponds to an evanescent excitation and inflates
$\eNR$ spuriously. For the uniform single-excitation W state, \eqref{eq:G1} reduces
exactly to $|AF|^2$, recovering the classical limit. The Dicke ladder remains the
correct and efficient representation for population dynamics and superradiance,
where no directional phase is involved. The second-order coherence, which has no classical analogue, is
\begin{equation}
\Gtwo(\hat{\mathbf{r}}_1,\hat{\mathbf{r}}_2)=
\Trace\!\bigl[\Aop^{\dagger}_1\Aop^{\dagger}_2\Aop_2\Aop_1\rho\bigr].
\label{eq:G2}
\end{equation}

\subsection{Quantum States}
Four canonical states are studied here. The W and GHZ families follow the standard definitions of D\"ur \textit{et al.}~\cite{Dur} and Greenberger \textit{et al.}~\cite{GHZref}. For the 1-D ULA ($N=4$) and 2-D UPA
($MN=16$) respectively:
\begin{align}
\ketv{W}&=\frac{1}{\sqrt{N}}\sum_{j}\ketv{e_j},\qquad\text{\cite{Dur}}
\label{eq:W}\\
\ketv{\Psi_{\mathrm{TD}}}&=\frac{1}{\sqrt{N}}\sum_{j}
e^{-jk_0\,\hat{\mathbf{r}}_0\cdot\mathbf{r}_j}\ketv{e_j},
\label{eq:TD}\\
\ketv{\Phi_{\mathrm{GHZ}}}&=\tfrac{1}{\sqrt2}
\bigl(\ketv{g\cdots g}+\ketv{e\cdots e}\bigr)\qquad\text{\cite{GHZref}}.
\label{eq:GHZ}
\end{align}
The W state is phase-free and must return $\eNR=0$; it is the framework
validation. The timed Dicke state embeds the propagation phase as an
\emph{intrinsic property of the quantum state}, not as an external feed
excitation. For the 1-D array with $d_{\mathrm{NR}}=0.3\lambda_0$, with $\hat{\mathbf{r}}_0$ along the array axis the phase step is
$k_0d_{\mathrm{NR}}=108^\circ$, placing element phases at
$0^\circ,-108^\circ,-216^\circ,-324^\circ$: four points distributed around the
unit circle with no mirror symmetry. The negative sign is that
of~\eqref{eq:conv}, so the beam points towards $\hat{\mathbf{r}}_0$. For the 2-D array the analogous state carries
\emph{two} independent gradients,
\begin{equation}
c_{mn}=\frac{1}{\sqrt{MN}}
\exp\!\bigl[-jk_0\bigl(m\,d_{\mathrm{NR},x}u_0+n\,d_{\mathrm{NR},y}v_0\bigr)\bigr],
\label{eq:TD2D}
\end{equation}
with $u_0=\sin\theta_0\cos\phi_0$ and $v_0=\sin\theta_0\sin\phi_0$ the direction
cosines of the intended beam direction $\hat{\mathbf{r}}_0$, which must satisfy
$u_0^2+v_0^2\leq1$. Throughout we take
$\hat{\mathbf{r}}_0=(\theta_0,\phi_0)=(90^\circ,45^\circ)$, the diagonal
end-fire direction.\\
$\ketv{\Psi_{\mathrm{TD}}}$ is treated as an
idealized prepared state, but its phase structure is not arbitrary. It is the profile a single far-field photon incident from $\hat{\mathbf{r}}_0$ would
imprint on the array, by the same argument that gives~\eqref{eq:drive}.
Conceptually the sequence is: (i)~place single excitation in the array, for
instance by weak resonant excitation from the direction $\hat{\mathbf{r}}_0$ or
by site-selective optical control; (ii)~imprint the spatial phase profile
$\varphi_j=-k_0\hat{\mathbf{r}}_0\!\cdot\!\mathbf{r}_j$, which for a
travelling-wave drive arises automatically from propagation delay;
(iii)~let the array evolve under~\eqref{eq:lindblad_dd}; (iv)~detect the emitted
photon in the far field. We emphasise that state-preparation fidelity, control
bandwidth, and the achievable phase accuracy lie outside the present model;
Sec.~\ref{sec:limits} quantifies how the results degrade once phase errors are
admitted. The GHZ state carries no single-excitation \emph{coherence} between distinct
sites. As shown in Sec.~\ref{sec:ghz}, this does not suppress $\Gone$: the
diagonal populations survive and give an isotropic $\Gone=N/2$, together with
$\Gtwo\neq0$ and $g^{(2)}=3$.

\subsection{Directional Asymmetry Metrics}
\label{sec:eta2d}
For the 1-D ULA the spatial-reversal partner of $\theta$ is $\pi-\theta$, which flips
$\cos\theta$, and
\begin{equation}
\eNR^{1\mathrm{D}}=\frac{\displaystyle\int_0^{\pi}
\bigl|\Gone(\theta)-\Gone(\pi-\theta)\bigr|\sin\theta\,d\theta}
{\displaystyle\int_0^{\pi}\Gone(\theta)\sin\theta\,d\theta}.
\label{eq:eta1d}
\end{equation}
The $\sin\theta$ weight is the physical solid-angle measure, identical to that
used in the planar case below, so the 1-D and 2-D values are directly
comparable. It is observed that the unweighted slice metric used in earlier work gives
$1.791$ for the same state; the solid-angle value is $1.791$.
For the 2-D UPA this substitution \emph{fails}: the pattern depends on
$(\theta,\phi)$ only through the direction cosines
$u=\sin\theta\cos\phi$, $v=\sin\theta\sin\phi$, and since
$\sin(\pi-\theta)=\sin\theta$, the map $\theta\rightarrow\pi-\theta$ leaves
$(u,v)$ unchanged --- it is the identity on the pattern and yields $\eNR\equiv0$
for \emph{every} state. The physical reversal $\hat{\mathbf{r}}\rightarrow-\hat{\mathbf{r}}$
requires $(u,v)\rightarrow(-u,-v)$, i.e.\ $\phi\rightarrow\phi+\pi$. Hence, it can be written that
\begin{equation}
\eNR^{2\mathrm{D}}=
\frac{\displaystyle\int_0^{2\pi}\!\!\int_0^{\pi}
\bigl|\Gone(\theta,\phi)-\Gone(\theta,\phi{+}\pi)\bigr|\sin\theta\,d\theta\,d\phi}
{\displaystyle\int_0^{2\pi}\!\!\int_0^{\pi}
\Gone(\theta,\phi)\sin\theta\,d\theta\,d\phi}.
\label{eq:eta2d}
\end{equation}
Numerically, the $\phi$ grid must span $[0,2\pi)$ with an \emph{even} number of
samples and exclude the duplicate endpoint; otherwise the discrete half-rotation
is off by one sample and a uniform-phase state spuriously returns
$\eNR\approx10^{-2}$ instead of $10^{-16}$.

\subsection{Collective Emission Rate}
In the Dicke limit $d\ll\lambda_0$ the symmetric single-excitation state decays
at $N\Gamma$. Beyond that limit the decay \emph{eigenrates} are the eigenvalues
of $\Gamma_{ij}$, while the initial decay rate of a particular state
$\mathbf{c}$ is $\mathbf{c}^{\dagger}\bm{\Gamma}\mathbf{c}$. This is developed
in Sec.~\ref{sec:decay}, where we also show why the frequently quoted
$N^{2}\Gamma$ does not apply to the timed Dicke state.

\section{What Is, and What Is Not, Quantum}
\subsection{A Classical-Equivalence Lemma}
\label{sec:lemma}
\begin{lemma}
\label{lem:ce}
For any single-excitation state with amplitudes $c_j$, the Glauber first-order
coherence equals the power pattern of a classical array excited by the
identical complex weights $w_j=c_j$:
\begin{equation}
\Gone(\hat{\mathbf{r}})
=\Bigl|\textstyle\sum_j c_j\,e^{+jk_0\hat{\mathbf{r}}\cdot\mathbf{r}_j}\Bigr|^2
=\bigl|AF(\hat{\mathbf{r}})\bigr|^2_{w=c}.
\end{equation}
\end{lemma}
The proof is immediate from~\eqref{eq:G1pure}. The consequence is important: the
timed Dicke pattern is reproduced exactly by a classical endfire phased array
with progressive phase $\beta=-k_0d$. We verified this numerically, obtaining
$\max_{\hat{\mathbf{r}}}\bigl|\Gone_{\mathrm{TD}}(\hat{\mathbf{r}})
-|AF(\hat{\mathbf{r}})|^2_{w=c}\bigr|=4\times10^{-16}$, i.e.\ agreement to
machine precision, with $\eNR=1.791$ and $D_0=6.73$~dBi for \emph{both}.\\
We will keep two tiers apart for the rest of the paper. The first is
\emph{classical radiation-pattern equivalence}: the fact that
$\Gone_{\mathrm{TD}}$ and $|AF|^2$ coincide, so that anything a first-order
intensity measurement can see is classically reproducible. The second is
\emph{nonclassical photon statistics}, which no arrangement of classical
currents can imitate. The asymmetry parameter lives entirely in the first tier;
the certificate of Sec.~\ref{sec:cert} lives in the second. Conflating the two
is the error this paper exists to remove.\\
It follows that the frequently repeated statement ``all passive classical
antennas satisfy $\eNR=0$'' is false as written. Lorentz reciprocity constrains
the \emph{transmit--receive} pair, $S_{21}=S_{12}$; it does \emph{not} force
pattern symmetry $G(\hat{\mathbf{r}})=G(-\hat{\mathbf{r}})$. Any endfire array
breaks the latter while fully obeying the former. Pattern asymmetry is therefore
a necessary but not sufficient indicator of non-classicality.

\subsection{A Tight Bound on Pattern Asymmetry}
\begin{theorem}
\label{thm:bound}
For any state and any measure-preserving mirror $\mathcal{M}$,
$\;0\leq\eNR\leq2$, with $\eNR=2$ if and only if $\Gone$ and
$\Gone\!\circ\!\mathcal{M}$ have disjoint support.
\end{theorem}
\begin{IEEEproof}
$\mathcal{M}$ is an involution of the sphere preserving the solid-angle measure,
so $\int\Gone\!\circ\!\mathcal{M}\,d\Omega=\int\Gone\,d\Omega$. By the triangle
inequality,
$\int|\Gone-\Gone\!\circ\!\mathcal{M}|\,d\Omega\leq
\int\Gone\,d\Omega+\int\Gone\!\circ\!\mathcal{M}\,d\Omega=2\!\int\!\Gone\,d\Omega$.
Dividing by $\int\Gone\,d\Omega$ gives $\eNR\leq2$; equality requires
$\min(\Gone,\Gone\!\circ\!\mathcal{M})=0$ pointwise.
\end{IEEEproof}
The measured values are thus best read as fractions of an absolute ceiling:
$1.791/2=89.5\%$ in 1-D and $1.867/2=93.3\%$ in 2-D. The planar array is closer
to the maximum because a two-axis phase gradient can null the mirror direction
over a two-dimensional neighbourhood rather than a single angle.

\subsection{The Nonclassical Directional-Emission Certificate}\label{sec:IV-C}
\label{sec:cert}
The preceding section settled the first tier: at the level of
\emph{classical radiation-pattern equivalence}, the timed Dicke array offers
nothing a feed network cannot. We now turn to the second tier, where the
situation reverses. The genuinely nonclassical resource is photon statistics. The normalised
second-order coherence,
\begin{equation}
g^{(2)}(\hat{\mathbf{r}}_1,\hat{\mathbf{r}}_2)=
\frac{\Gtwo(\hat{\mathbf{r}}_1,\hat{\mathbf{r}}_2)}
{\Gone(\hat{\mathbf{r}}_1)\,\Gone(\hat{\mathbf{r}}_2)},
\label{eq:g2norm}
\end{equation}
satisfies $g^{(2)}\geq1$ for every classical field, since classical intensities
obey $\langle I^2\rangle\geq\langle I\rangle^2$. We therefore propose:

\begin{definition}[Nonclassical Directional-Emission Certificate]
The conjunction
\begin{equation}
\eNR>0 \quad\text{and}\quad
g^{(2)}(\hat{\mathbf{r}}_0,\hat{\mathbf{r}}_0)<1
\label{eq:cert}
\end{equation}
is a \emph{sufficient} certificate for nonclassical directional emission, where
$\hat{\mathbf{r}}_0$ is the peak-emission direction.
\end{definition}
Two qualifications matter in this specific use case. Firstly, the condition is sufficient, not necessary. As states with $\eNR=0$ may still be strongly nonclassical, so~\eqref{eq:cert} is not a
definition of quantumness. Secondly, $g^{(2)}$ must be evaluated where
$\Gone>0$; at a pattern null $\Gone\rightarrow0$ and
$g^{(2)}=\Gtwo/[\Gone]^2$ becomes numerically undefined. Evaluating at
$\hat{\mathbf{r}}_0$ removes both difficulties. The two conditions draw on distinct physical resources: pattern asymmetry can
be produced by a classical phased array, whereas $g^{(2)}<1$ requires a
nonclassical emitter. Their conjunction therefore identifies nonclassical
directional emission, and a classical phased array cannot supply it because it
has $g^{(2)}\geq1$ at every angle. Direct
computation in the full $2^{N}$ Hilbert space gives, for the timed Dicke and W
states, $g^{(2)}(\hat{\mathbf{r}},\hat{\mathbf{r}})=0.0000$ exactly, because a
one-excitation state cannot emit two photons. The timed Dicke array therefore
satisfies~\eqref{eq:cert} with $\eNR=1.791$ and $g^{(2)}=0$: it radiates with the
directional asymmetry of a phased array while emitting strictly one photon at a
time. This is the precise sense in which the effect is quantum, and it matches
the emphasis of Slepyan and Boag on single-photon pumping~\cite{SlepyanBoag}.\\
For equivalent experimental interpretation, the certificate is stated so that it maps onto a concrete measurement. With a
far-field detector and a Hanbury~Brown--Twiss pair:
\emph{(1)}~measure the first-order intensity at the peak direction
$\hat{\mathbf{r}}_0$; \emph{(2)}~measure it at the mirror direction
$-\hat{\mathbf{r}}_0$; \emph{(3)}~repeat over the sphere and form $\eNR$
from~\eqref{eq:eta2d}; \emph{(4)}~place the HBT pair at $\hat{\mathbf{r}}_0$;
\emph{(5)}~record the normalised coincidence rate $g^{(2)}(0)$;
\emph{(6)}~if $\eNR>0$ and $g^{(2)}(0)<1$, the emission satisfies the
certificate. Only step~(3) requires an angular scan; steps~(4)--(5) are a
standard single-emitter antibunching measurement. In practice detector dark
counts, background photons and finite timing resolution raise the measured
$g^{(2)}(0)$ above its ideal value, so the certificate should be read as a
threshold test with the usual background subtraction, not as a demand for
$g^{(2)}=0$ exactly. The GHZ state occupies the opposite corner: $\Gone$ is isotropic (below) while $g^{(2)}=3$ from~\eqref{eq:g2ghz} --- a bunched multi-photon emitter with
no directional single-photon channel.

\subsection{The GHZ Pattern Is Isotropic, Not Vanishing}
\label{sec:ghz}
It is often stated that the GHZ state gives $\Gone\approx0$, which is not necessarily always correct.
Writing $\Gone=\sum_{ij}e^{\,jk_0\hat{\mathbf{r}}\cdot(\mathbf{r}_i-\mathbf{r}_j)}
\langle\hat{\sigma}^{+}_{i}\hat{\sigma}^{-}_{j}\rangle$, the GHZ state has no
single-excitation coherence between distinct sites, so all $i\neq j$ terms
vanish, while each diagonal term contributes $\langle\hat{n}_i\rangle=1/2$. Hence, it becomes
\begin{equation}
\Gone_{\mathrm{GHZ}}(\hat{\mathbf{r}})=\frac{N}{2}\qquad\text{for all }\hat{\mathbf{r}}.
\label{eq:ghziso}
\end{equation}
The two-photon term follows from the same argument: only the $\ketv{e\cdots e}$
branch contributes, and it supplies $N(N-1)$ ordered pairs, so
$\Gtwo_{\mathrm{GHZ}}=N(N-1)$ in the same spatial mode. Hence
\begin{equation}
g^{(2)}_{\mathrm{GHZ}}=\frac{N(N-1)}{(N/2)^2}=\frac{4(N-1)}{N},
\label{eq:g2ghz}
\end{equation}
which equals $3$ for $N=4$, exceeding the thermal value $2$. Both closed forms
were verified numerically in the full $2^N$ space. The GHZ array
radiates \emph{incoherently and isotropically at half the single-QD rate per
site}, not weakly. Its directivity is $0$~dBi because the pattern is uniform, not
because the emission is suppressed.

\section{Collective Decay with Dipole--Dipole Coupling}
\label{sec:decay}
The Lindblad equation~\eqref{eq:lindblad} written with a single collective
$\hat{J}^{-}$ presupposes the Dicke limit $d\ll\lambda_0$, in which every pair
decays cooperatively at the same rate. At $d=0.3\lambda_0$ this is not the case.
The correct open-system generator for extended arrays is~\cite{Breuer,Masson}
\begin{align}
\dot{\rho}=-\frac{i}{\hbar}\Bigl[\hat{H}_0+\hbar\!\!\sum_{i\neq j}\!\!J_{ij}
\hat{\sigma}^{+}_i\hat{\sigma}^{-}_j,\;\rho\Bigr]
+\sum_{i,j}\Gamma_{ij}\Bigl(&\hat{\sigma}^{-}_j\rho\hat{\sigma}^{+}_i \nonumber\\
-\tfrac12\bigl\{\hat{\sigma}^{+}_i\hat{\sigma}^{-}_j,&\,\rho\bigr\}\Bigr),
\label{eq:lindblad_dd}
\end{align}
where, for dipoles perpendicular to the separation axis and
$\xi_{ij}=k_0|\mathbf{r}_i-\mathbf{r}_j|$,
\begin{align}
\frac{\Gamma_{ij}}{\Gamma}&=\frac{3}{2}
\left[\frac{\sin\xi}{\xi}+\frac{\cos\xi}{\xi^{2}}-\frac{\sin\xi}{\xi^{3}}\right],
\label{eq:Gij}\\
\frac{J_{ij}}{\Gamma}&=\frac{3}{4}
\left[-\frac{\cos\xi}{\xi}+\frac{\sin\xi}{\xi^{2}}+\frac{\cos\xi}{\xi^{3}}\right].
\label{eq:Jij}
\end{align}
The coherent term $J_{ij}$ is the collective Lamb shift. At the design point
$d=0.3\lambda_0$ ($\xi=1.885$), the nearest-neighbour values are
$\Gamma_{ij}=0.413\Gamma$ and $J_{ij}=0.289\Gamma$: the Lamb shift is $29\%$ of
the decay rate and cannot be discarded.

\subsection{Correction to the $N^{2}\Gamma$ Superradiance Claim}
Two separate statements are frequently conflated. Dicke's $N^{2}$ law is the
\emph{peak burst rate of a fully inverted ensemble} at the centre of its cascade.
The timed Dicke state, by contrast, carries exactly \emph{one} excitation and can
emit exactly one photon; its decay rate in the Dicke limit is $N\Gamma$. Beyond
that limit the decay rates are the eigenvalues of $\Gamma_{ij}$, and the
superradiant mode decays at
\begin{equation}
\Gamma_{\mathrm{SR}}=\lambda_{\max}\bigl(\Gamma_{ij}\bigr),
\label{eq:lammax}
\end{equation}
which tends to $N\Gamma$ only as $d\rightarrow0$. Crucially,
$\lambda_{\max}$ is the \emph{maximum collective rate the array can support},
not the rate of any particular state. A timed Dicke state
$\mathbf{c}$ decays initially at its projection onto the coupling matrix,
\begin{equation}
\Gamma_\psi(0)=\mathbf{c}^{\dagger}\bm{\Gamma}\,\mathbf{c},
\label{eq:rateproj}
\end{equation}
and, since $\mathbf{c}$ is not generally an eigenvector of $\bm{\Gamma}$, the
subsequent decay is multi-exponential. For the 1-D chain at $d=0.3\lambda_0$ the
timed Dicke projection onto the four eigenmodes is
$(0.118,\,0.550,\,0.075,\,0.257)$, giving
$\Gamma_\psi(0)=1.01\Gamma$ against $\lambda_{\max}=1.61\Gamma$; for the
$4\times4$ array $\Gamma_\psi(0)=2.08\Gamma$ against
$\lambda_{\max}=2.75\Gamma$. Table~\ref{tab:decay} gives the
computed values. At $d=0.3\lambda_0$ the timed Dicke states decay at $1.01\Gamma$ and
$2.08\Gamma$; quoting $16\Gamma$ and $256\Gamma$ overstates these by more than
an order of magnitude.

\begin{table}[!t]
\renewcommand{\arraystretch}{1.18}
\caption{Superradiant Rate $\lambda_{\max}(\Gamma_{ij})$ in Units of $\Gamma$}
\label{tab:decay}
\centering
\begin{tabular}{lccccc}
\toprule
Array & $d{=}0.01\lambda_0$ & $0.1\lambda_0$ & $0.3\lambda_0$ & $0.5\lambda_0$ & $N\Gamma$\\
\midrule
1-D, $N{=}4$      & $3.99$  & $3.32$ & $1.61$ & $1.29$ & $4$\\
2-D, $4{\times}4$ & $15.94$ & $10.84$ & $2.75$ & $2.10$ & $16$\\
\midrule
\multicolumn{6}{@{}p{0.94\columnwidth}@{}}{Dicke limit recovered as $d\rightarrow0$;
 $N^{2}\Gamma$ is never attained.}\\
\multicolumn{6}{@{}p{0.94\columnwidth}@{}}{Subradiant mode, 2-D at $0.3\lambda_0$:
 $\lambda_{\min}=0.0010\,\Gamma$.}\\
\bottomrule
\end{tabular}
\end{table}

\subsection{Subradiant Dark States}
The same eigenanalysis exposes the complementary regime. For the $4\times4$
array at $d=0.3\lambda_0$ the smallest eigenvalue is
$\lambda_{\min}=0.0010\,\Gamma$, a lifetime enhancement of three orders of
magnitude. These dark states are invisible in a purely collective $\hat{J}^{-}$
model, in which the $N-1$ non-symmetric modes are exactly decoupled rather than
weakly radiating. This radiatively subradiant eigenmode suggests a potential photon-storage
mechanism, subject to nonradiative and dephasing losses, which are not included
here and which would in practice bound the achievable lifetime well below
$1000/\Gamma$. To the best of our knowledge such modes have not been reported
for QD antenna arrays of this kind.

\section{Quantum Direction-of-Arrival Estimation}
This section extends the previous discussion of QD arrays for direction-of-arrival (DoA) estimation applications.
\subsection{Quantum Spatial Coherence Matrix}
 Classical DoA estimator, MUSIC estimates a sample covariance
$\mathbf{R}=\mathbb{E}[\mathbf{x}\mathbf{x}^{\mathsf{H}}]$ from $L$ temporal
snapshots~\cite{Schmidt}. The quantum counterpart is defined directly in the
\emph{array-element basis}, and is exactly the correlation matrix already
introduced in~\eqref{eq:G1}:
\begin{equation}
[\mathbf{R}_Q]_{ij}=C_{ij}
=\Trace\!\bigl(\hat{\sigma}^{+}_i\hat{\sigma}^{-}_j\rho\bigr),
\qquad i,j=1,\dots,MN .
\label{eq:QSCM}
\end{equation}
so $\mathbf{R}_Q=\mathbf{C}$ exactly, with
$\mathbf{R}_Q\in\mathbb{C}^{MN\times MN}$ and diagonal entries equal to the
individual emitter excitation probabilities. The first-order coherence is the
same quadratic form as~\eqref{eq:G1},
\begin{equation}
\Gone(\hat{\mathbf{r}})
=\mathbf{a}^{\mathsf{H}}(\hat{\mathbf{r}})\,\mathbf{R}_Q\,
\mathbf{a}(\hat{\mathbf{r}}),
\label{eq:G1RQ}
\end{equation}
with $\mathbf{a}(\theta,\phi)\in\mathbb{C}^{MN}$ from~\eqref{eq:steer}. Both
objects therefore live in the same $MN$-dimensional element space, and the
signal/noise subspace decomposition below is the standard array one. (Defining
$\mathbf{R}_Q$ instead on a grid of observation directions would give a matrix
whose dimension is the grid size, which is incompatible with an
$MN$-dimensional steering vector.)

\subsection{Source Model: How the Sources Enter $\mathbf{R}_Q$}
\label{sec:srcmodel}
For direction finding the array acts as a \emph{receiver}, so we must specify
how $K$ incident sources produce $\rho$ and hence $\mathbf{R}_Q$. Consider $K$
mutually incoherent far-field sources at directions
$(\theta_k,\phi_k)$ driving the array weakly, so that the response stays in the
single-excitation manifold. Source $k$, with complex amplitude $E_k$, imprints
its own propagation phase across the array and excites the amplitude vector
\begin{equation}
\mathbf{c}^{(k)}=E_k\,\mathbf{a}^{*}(\theta_k,\phi_k),
\label{eq:drive}
\end{equation}
the conjugate appearing because an incoming wave imprints
$e^{-jk_0\hat{\mathbf{r}}_k\cdot\mathbf{r}_j}$, consistent
with~\eqref{eq:conv}. Since $C_{ij}=c_i^{*}c_j$, a single source therefore
contributes $\mathbf{C}=|E_k|^2\,
\mathbf{a}(\theta_k,\phi_k)\mathbf{a}^{\mathsf{H}}(\theta_k,\phi_k)$. Because
the sources are mutually incoherent, their cross terms average to zero over the
independent phases, and the element-basis coherence matrix becomes
\begin{equation}
\mathbf{R}_Q=\sum_{k=1}^{K}p_k\,
\mathbf{a}(\theta_k,\phi_k)\mathbf{a}^{\mathsf{H}}(\theta_k,\phi_k)
+\sigma_n^2\mathbf{I},
\label{eq:srcmodel}
\end{equation}
with $p_k=\langle|E_k|^2\rangle$ and $\sigma_n^2$ the background level. We
confirmed~\eqref{eq:srcmodel} numerically by averaging $2\times10^{5}$
independent phase realisations of~\eqref{eq:drive}, which reproduces the
closed form to a relative error of $8\times10^{-4}$.\\
Equation~\eqref{eq:srcmodel} makes the rank structure explicit and resolves an
issue that a state-only description leaves open: a single pure state
$\rho=\ketv{\psi}\brav{\psi}$ gives a rank-one $\mathbf{R}_Q$ and cannot
represent two independent sources. Under~\eqref{eq:srcmodel} the noiseless
coherence matrix has $\operatorname{rank}(\mathbf{R}_Q)=K$; for the two-source
case studied below the eigenvalues are $(17.07,\,14.93,\,\mathcal{O}(10^{-15}),
\dots)$, i.e.\ exactly two significant modes, so the noise subspace is
well defined.

\subsection{Quantum MUSIC}
Eigendecomposing
$\mathbf{R}_Q=\sum_k\lambda_k\ketv{v_k}\brav{v_k}$ with
$\lambda_1\!\geq\!\cdots\!\geq\!\lambda_{MN}$, the $K$ largest eigenvalues span the signal subspace and the remainder the
noise subspace $\mathbf{E}_n$. For an exact model $\rho$ one simply has
$K=\operatorname{rank}(\mathbf{R}_Q)$. Classical order-selection rules such as
MDL~\cite{Wax} presuppose a \emph{sample} covariance and therefore apply only to
the finite-shot estimator
$\widehat{\mathbf{R}}_Q=N_s^{-1}\sum_{s}\mathbf{R}_Q^{(s)}$
formed from $N_s$ repeated preparations. The
Q-MUSIC pseudo-spectrum is
\begin{equation}
P_{\mathrm{QMUSIC}}(\theta,\phi)=
\frac{1}{\mathbf{a}^{\mathsf{H}}(\theta,\phi)\,
\mathbf{E}_n\mathbf{E}_n^{\mathsf{H}}\,\mathbf{a}(\theta,\phi)},
\label{eq:qmusic}
\end{equation}
whose $K$ peaks on the $(\theta,\phi)$ sphere locate the sources.
Algorithm~\ref{alg:qmusic} summarises the procedure.

\begin{algorithm}[!t]
\caption{Quantum MUSIC DoA for the $M\times N$ QD-UPA}
\label{alg:qmusic}
\begin{algorithmic}[1]
\Require $\mathbf{R}_Q$ from \eqref{eq:srcmodel} (Sec.~\ref{sec:srcmodel}), geometry $\{\mathbf{r}_{mn}\}$, grid
         $\{(\theta_i,\phi_j)\}$
\Ensure  $\{(\hat\theta_k,\hat\phi_k)\}_{k=1}^{K}$
\State Build $\mathbf{R}_Q$ from \eqref{eq:QSCM}
\State $[\lambda_k,\ketv{v_k}]\leftarrow\mathrm{eig}(\mathbf{R}_Q)$, sorted descending
\State $K\leftarrow\operatorname{rank}(\mathbf{R}_Q)$ (exact $\rho$), or MDL on $\widehat{\mathbf{R}}_Q$ (finite shots)
\State $\mathbf{E}_n\leftarrow[\ketv{v_{K+1}},\dots,\ketv{v_{MN}}]$
\For{each grid point $(\theta_i,\phi_j)$}
  \State $\mathbf{a}\leftarrow$ steering vector \eqref{eq:steer}
  \State $P_{\mathrm{QMUSIC}}\leftarrow
         1/(\mathbf{a}^{\mathsf{H}}\mathbf{E}_n\mathbf{E}_n^{\mathsf{H}}\mathbf{a})$
\EndFor
\State \Return the $K$ largest well-separated peaks
\end{algorithmic}
\end{algorithm}

\subsection{Quantum Fisher Information and a Bounded Advantage}
\label{sec:qfi}
For a pure single-excitation probe $\ketv{\psi}=\sum_{mn}c_{mn}\ketv{e_{mn}}$,
the elevation angle $\theta$ is imprinted through the diagonal generator
\begin{equation}
\hat{G}_\theta=k_0\bigl(m\,d_x\,u' + n\,d_y\,v'\bigr),
\quad u'=\partial_\theta u,\;\; v'=\partial_\theta v,
\label{eq:generator}
\end{equation}
with $u'=\cos\theta\cos\phi$ and $v'=\cos\theta\sin\phi$. The QFI is
\begin{align}
\FQ(\theta)&=4\,\Var\!\bigl(\hat{G}_\theta\bigr) \nonumber\\
&=4\left[\sum_{mn}|c_{mn}|^{2}G_{mn}^{2}
-\left(\sum_{mn}|c_{mn}|^{2}G_{mn}\right)^{\!2}\right],
\label{eq:qfi}
\end{align}
and the QCRB over $N_{\mathrm{meas}}$ repetitions is
$\Delta\theta\geq1/\sqrt{N_{\mathrm{meas}}\FQ}$~\cite{Helstrom,Liu}.

\textbf{Symmetric W-state probe:}
As a reference case, we first consider the symmetric W state, in which the
single excitation is distributed uniformly over all $M\times M$ elements.
Thus, $|c_{mn}|^2=1/M^2$, with $m$ and $n$ statistically independent.
For this uniform spatial distribution,
$\mathrm{Var}(m)=\mathrm{Var}(n)=(M^2-1)/12$, which leads to
\begin{equation}
\FQ^{W}=\tfrac{1}{3}k_0^2d^2\bigl(u'^2+v'^2\bigr)\bigl(M^2-1\bigr).
\label{eq:Fsep}
\end{equation}

\textbf{Entangled corner-pair probe:}
To examine the effect of maximum spatial separation on the attainable
angular information, we next consider an entangled probe formed by the
two maximally separated corner elements. The excitation is prepared in
the superposition
$\frac{1}{\sqrt{2}}(|e_{00}\rangle+|e_{M-1,M-1}\rangle)$.
For this probe, the generator takes two equally likely values
$\{0,k_0d(M-1)(u'+v')\}$. Hence,
\begin{equation}
\FQ^{\mathrm{corner}}=k_0^2d^2(M-1)^2\bigl(u'+v'\bigr)^2.
\label{eq:Fent}
\end{equation}

Dividing~\eqref{eq:Fent} by~\eqref{eq:Fsep} and specialising to the diagonal
look direction $\phi=45^\circ$, where $u'=v'$ and thus
$(u'+v')^2=2(u'^2+v'^2)$, gives the central closed-form result
\begin{equation}
\frac{\FQ^{\mathrm{corner}}}{\FQ^{W}}
=\frac{3(M-1)(u'+v')^2}{(M+1)(u'^2+v'^2)}
\;\xrightarrow{\;\phi=45^\circ\;}\;
\frac{6(M-1)}{M+1}.
\label{eq:ratio}
\end{equation}
Equation~(39) provides a single-parameter comparison of
$F_{\theta\theta}$ for the two probes at fixed azimuth. Therefore, this
result characterizes elevation estimation for a specified azimuth and
does not by itself describe the complete two-dimensional DoA problem.
For simultaneous estimation of elevation and azimuth, the full QFI
matrix $F_Q$ is required, with entries $F_{\alpha\beta}$,
$\alpha,\beta\in\{\theta,\phi\}$, accounting for the information and
correlations associated with both angular parameters. In particular,
the corner-pair probe must be examined through this matrix because its
geometry can lead to a loss of sensitivity to one of the angular
coordinates.
Equation~\eqref{eq:ratio} concerns \emph{single-parameter elevation estimation
at fixed azimuth}. The full two-parameter problem requires the QFI matrix $\FQM$ with entries
$F_{\alpha\beta}$, $\alpha,\beta\in\{\theta,\phi\}$, and
\begin{equation}
\mathrm{Cov}(\hat{\bm{\theta}})\succeq
\bigl(N_{\mathrm{meas}}\FQM\bigr)^{-1}.
\label{eq:qcrbmat}
\end{equation}
For simultaneous multiparameter estimation, the attainability of the
matrix QCRB additionally depends on measurement compatibility. Hence,
we use $F_Q^{-1}$ as an information-theoretic precision bound and do not
assume a measurement that jointly saturates the bound. Evaluated at $(30^\circ,45^\circ)$, the W probe gives
$F_{\theta\theta}=37.01$, $F_{\phi\phi}=12.34$ and $F_{\theta\phi}=0$,
whereas the corner-pair probe gives $F_{\theta\theta}=133.24$ but
$F_{\phi\phi}=0$. The corner pair lies along the $\phi=45^\circ$ diagonal, so
its generator is insensitive to azimuth and $\FQM$ is singular: \emph{the probe
that maximises elevation information is blind in azimuth.} Any practical scheme
must therefore interleave probe geometries. This is a substantive limitation of
the corner-pair construction and is not visible in the scalar analysis.
\begin{figure}[ht!]
\centering
\includegraphics[width=\columnwidth]{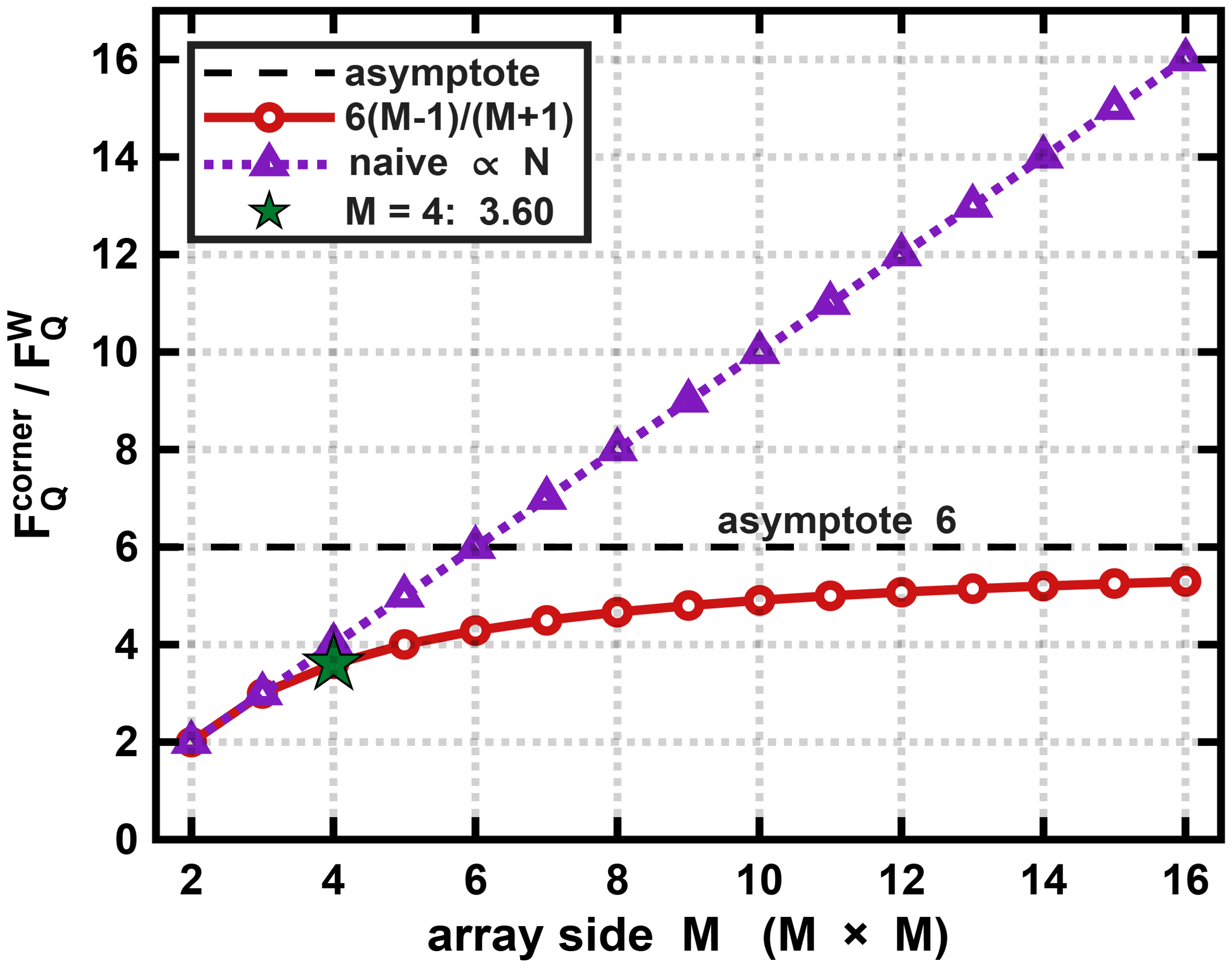}
\caption{The QFI enhancement saturates. Red: the closed form
$6(M-1)/(M+1)$ from~\eqref{eq:ratio}, approaching $6$ as $M\rightarrow\infty$.
Purple: the linear-in-$N$ scaling a naive Heisenberg argument would predict.
The star marks this work, $M=4$, ratio $3.60$ ($1.90\times$ in angular standard
deviation).}
\label{fig:qfisat}
\end{figure}

Equation~\eqref{eq:ratio} is the key single-parameter precision statement.
For $M=4$ it evaluates to $3.60$, an angular-precision gain of
$\sqrt{3.60}=1.90\times$; Fig.~\ref{fig:qfisat} plots the law against array
size and makes the saturation explicit. The ratio is
\emph{monotonically bounded}: as $M\rightarrow\infty$ it saturates
at $6$, i.e.\ at most $2.45\times$ in standard deviation. The result is specific
to an $M\times M$ uniform array, a corner-pair probe, $\phi=45^\circ$,
elevation estimation, and a pure-state model; it is not a universal
quantum-array law. Within that scope, \textbf{it is not Heisenberg $1/N$
scaling.} True
Heisenberg scaling requires $N$-fold phase accumulation on a single parameter;
a passive array of fixed aperture instead has its Fisher information bounded by
the geometric variance of the aperture, and entanglement can only redistribute
weight toward the aperture edges. Reports of $1/N$ DoA scaling for passive
quantum arrays should therefore be treated with care --- the aperture, not the
photon number, sets the bound.

\section{Results}
All results below were computed using standalone MATLAB scripting and independently cross-verified in Python. Both platforms agree up to at least six significant figures.

\subsection{Solver Validation}
Fig.~\ref{fig:pop} shows the Lindblad solution~\eqref{eq:lindblad} in the
17-dimensional Dicke subspace. The W state starts at $P_{\mathrm{exc}}(0)=1.000$
and decays faster than the isolated-QD reference $e^{-\Gamma t}$ --- the signature
of collective coupling through $\hat{J}^{\pm}$~\cite{Masson}. The GHZ state starts at exactly $P_{\mathrm{exc}}(0)=0.500$, as
\eqref{eq:GHZ} requires. We stress that this population value carries no
nonclassicality on its own --- an ordinary statistical mixture of excited and
unexcited emitters reproduces it. What is nonclassical is the coherence between
the two branches, and that coherence shows itself in the higher-order
correlations of Sec.~\ref{sec:cert}, not in the population curve.

\begin{figure}[!t]\centering
\includegraphics[width=\columnwidth]{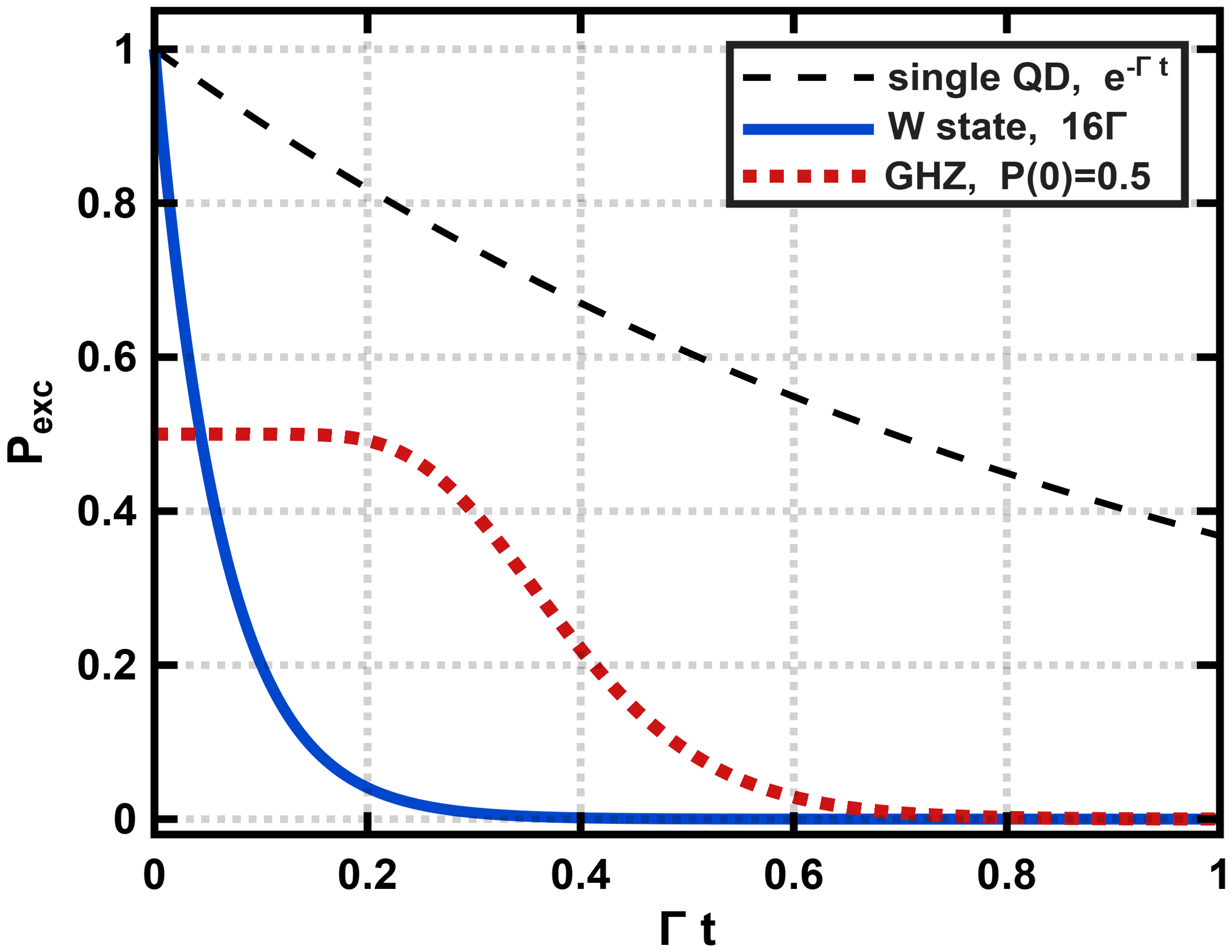}
\caption{Population dynamics from~\eqref{eq:lindblad}, solved in the symmetric
Dicke subspace of dimension $MN{+}1=17$. Blue: W state, $P_{\mathrm{exc}}(0)=1.000$,
decaying at $MN\Gamma=16\Gamma$ and so far faster than the dashed single-QD
reference $e^{-\Gamma t}$. Red: GHZ, starting at exactly $0.500$ and holding that
value until $\Gamma t\approx0.2$ before falling --- the delayed superradiant
cascade from the fully inverted branch. This figure is Model~A throughout; the
finite-spacing rates of Sec.~\ref{sec:decay} (Model~B) are shown separately in
Fig.~\ref{fig:super} and are not superimposed here.}
\label{fig:pop}
\end{figure}

\subsection{One-Dimensional Array}
Table~\ref{tab:1d} lists the 1-D metrics. The classical baseline reproduces
textbook values ($D_0=6.02$~dBi, HPBW $=26.3^\circ$). The timed Dicke state at
$d_{\mathrm{NR}}=0.3\lambda_0$ gives $\eNR=1.791$, matching the analytical
prediction of~\cite{SlepyanBoag}, with $\Gone(0^\circ)=1.0000$ and
$\Gone(180^\circ)=0.0625$ --- a front-to-back ratio of exactly $16.0\times$.
The classical phased-array baseline at the same spacing and with the same
complex weights reproduces $D_0=6.73$~dBi and $\eNR=1.791$ identically
(Table~\ref{tab:1d}), confirming the classical-equivalence lemma; the
distinguishing feature is therefore not the first-order pattern but the photon
statistics of Sec.~\ref{sec:cert}. The W state returns $\eNR=0.000$, confirming
that the formalism correctly identifies a phase-free state as symmetric.
Fig.~\ref{fig:polar1d} shows the corresponding polar patterns: the forward and
mirrored timed-Dicke lobes differ at every angle, whereas the classical and W
curves are indistinguishable from their own mirrors.

\begin{table}[!t]
\renewcommand{\arraystretch}{1.18}
\caption{One-Dimensional $N=4$ InGaAs QD-ULA}
\label{tab:1d}
\centering
\setlength{\tabcolsep}{4.5pt}
\begin{tabular}{lcccc}
\toprule
\textbf{Configuration} & $D_0$ & HPBW & SLL & $\eNR$\\
 & (dBi) & (deg) & (dB) & \\
\midrule
Classical, $0.5\lambda_0$ & $6.02$ & $26.3$ & $-11.30$ & $0.000$\\
W state, $0.5\lambda_0$ & $6.02$ & $26.3$ & $-11.30$ & $0.000$\\
Classical phased, $0.3\lambda_0$ & $6.73$ & $51.6$ & $-11.30$ & $1.791$\\
Timed Dicke, $0.3\lambda_0$ & $\mathbf{6.73}$ & $51.6$ & $-11.30$ & $\mathbf{1.791}$\\
GHZ, $0.5\lambda_0$ & $0.00$ & isotropic & --- & $0.000$\\
\midrule
\multicolumn{5}{@{}p{0.94\columnwidth}@{}}{Timed Dicke: $\Gone(0^\circ)=1.0000$,
 $\Gone(180^\circ)=0.0625$, F/B $=16.0\times$}\\
\multicolumn{5}{@{}p{0.94\columnwidth}@{}}{Phase step $k_0d_{\mathrm{NR}}=108^\circ$;
 optimum at $d=0.300\lambda_0$ ($\eNR=1.791$)}\\
\multicolumn{5}{@{}p{0.94\columnwidth}@{}}{Collective rate $\lambda_{\max}(\Gamma_{ij})=1.61\Gamma$ (not $16\Gamma$)}\\
\multicolumn{5}{@{}p{0.94\columnwidth}@{}}{$g^{(2)}=0$ exactly $\Rightarrow$ certificate~\eqref{eq:cert} satisfied}\\
\bottomrule
\end{tabular}
\end{table}

\begin{figure}[!t]\centering
\includegraphics[width=1\columnwidth]{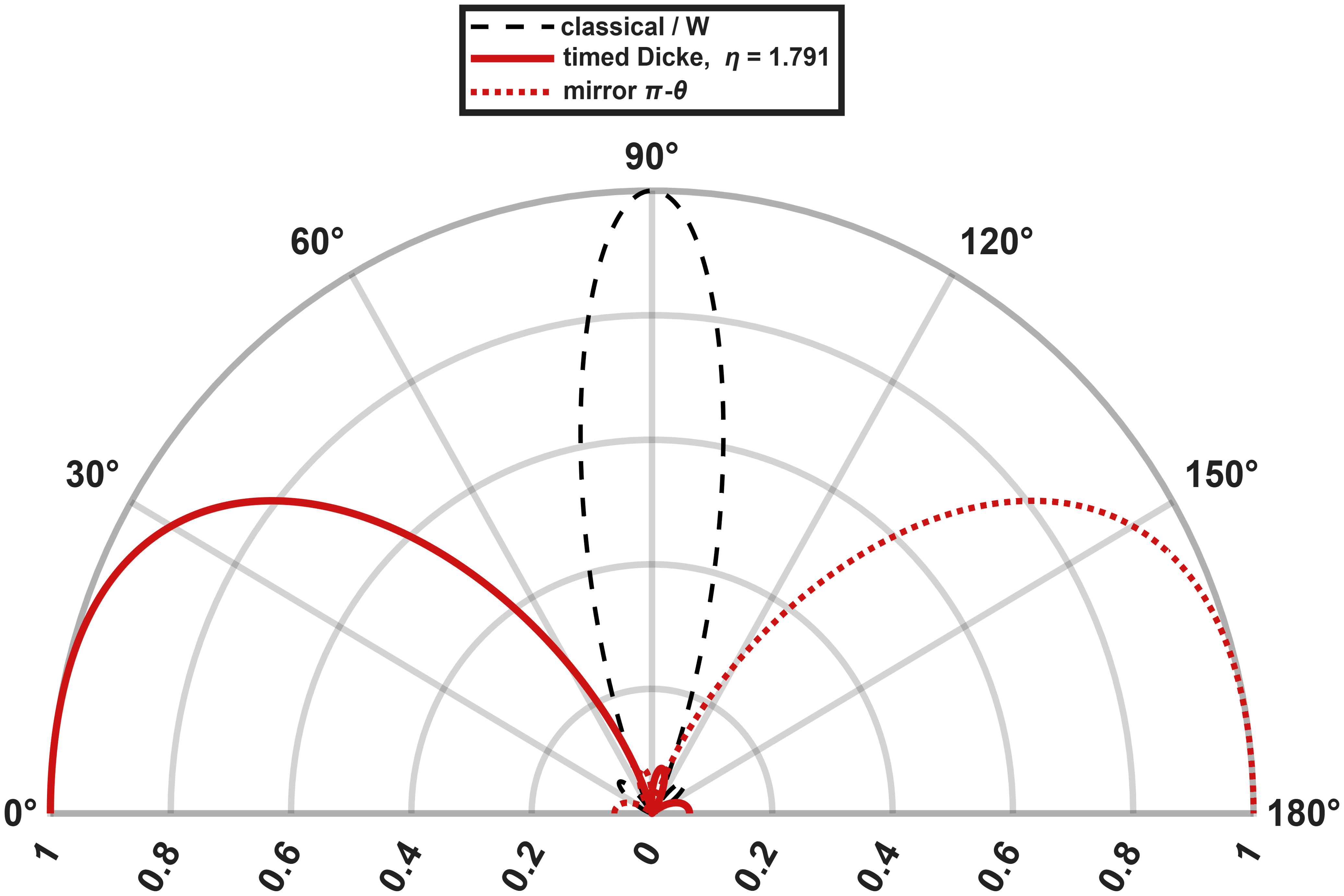}
\caption{1-D polar patterns. Dashed black: classical $|AF|^2$, symmetric,
$\eNR=0$. Blue: W state, also symmetric. Solid red: timed Dicke forward
$\Gone(\theta)$; dotted red: mirror $\Gone(\pi-\theta)$. The red pair differ at
every angle, directly demonstrating directional pattern asymmetry under spatial reversal.
By the classical-equivalence lemma (Lemma~\ref{lem:ce}) this asymmetry is
reproducible by a phased feed and is therefore not in itself evidence of broken
reciprocity; the nonclassical content of this state lies in its photon
statistics.}
\label{fig:polar1d}
\end{figure}

\subsection{Two-Dimensional Array}
Table~\ref{tab:2d} reports the planar results. Two validations bracket the
findings. First, the W state gives $\eNR^{2\mathrm{D}}=4\times10^{-16}$ ---
numerically exact zero, confirming both the formalism and the discretisation of
\eqref{eq:eta2d}. Second, the TD-2D state gives $\eNR^{2\mathrm{D}}=1.867$, a
$6.3\%$ increase over the 1-D value, since two independent phase gradients break
the symmetry over the full sphere rather than a single plane. The front-to-back
ratio reaches $27.2$~dB because the mirror direction lies in a deep
two-dimensional null, visible as the elevation cut in Fig.~\ref{fig:cut2d} and
over the full sphere in Fig.~\ref{fig:map2d}.\\
The TD-2D beam points at $(\theta,\phi)=(90^\circ,45^\circ)$: the diagonal
end-fire direction. This is the closest physically realisable point to the
formal optimum $u=v=1$, which is excluded by $u^2+v^2\leq1$. Consequently $D_0=9.49$~dBi is \emph{below} the classical broadside
$13.51$~dBi. This is the expected trade: the planar aperture is foreshortened at
end-fire, so pattern asymmetry is purchased at the cost of directivity --- a
trade-off absent in the 1-D case, where end-fire directivity exceeds broadside.

\begin{table}[!t]
\renewcommand{\arraystretch}{1.18}
\caption{Two-Dimensional $4\times4$ InGaAs QD-UPA}
\label{tab:2d}
\centering
\begin{tabular}{lcccc}
\toprule
\textbf{Configuration} & $D_0$ & HPBW$_\theta$ & $\eNR^{2\mathrm{D}}$ & Peak $(\theta,\phi)$\\
 & (dBi) & (deg) & & (deg)\\
\midrule
Classical, $0.5\lambda_0$ & $13.51$ & $13.2$ & $4\!\times\!10^{-16}$ & $(0,0)$\\
W state, $0.5\lambda_0$   & $13.51$ & $13.2$ & $4\!\times\!10^{-16}$ & $(0,0)$\\
TD-2D, $0.3\lambda_0$     & $9.49$ & $104.4$ & $\mathbf{1.867}$ & $(90,45)$\\
GHZ, $0.5\lambda_0$       & $0.00$ & isotropic & $0.000$ & uniform\\
\midrule
\multicolumn{5}{@{}p{0.94\columnwidth}@{}}{HPBW$_\theta$ is the elevation cut through the peak; for TD-2D the azimuth cut is $44.5^\circ$. Front-to-back $27.2$~dB.}\\
\multicolumn{5}{@{}p{0.94\columnwidth}@{}}{Optimum spacing $d=0.338\lambda_0$ ($\eNR^{2\mathrm{D}}=1.926$)}\\
\multicolumn{5}{@{}p{0.94\columnwidth}@{}}{Collective rate $\lambda_{\max}(\Gamma_{ij})=2.75\Gamma$; subradiant $\lambda_{\min}=0.001\Gamma$}\\
\multicolumn{5}{@{}p{0.94\columnwidth}@{}}{$\eNR^{2\mathrm{D}}=1.867=93.3\%$ of the bound $\eNR\leq2$ (Theorem~\ref{thm:bound})}\\
\bottomrule
\end{tabular}
\end{table}

\begin{figure}[!t]\centering
\includegraphics[width=\columnwidth]{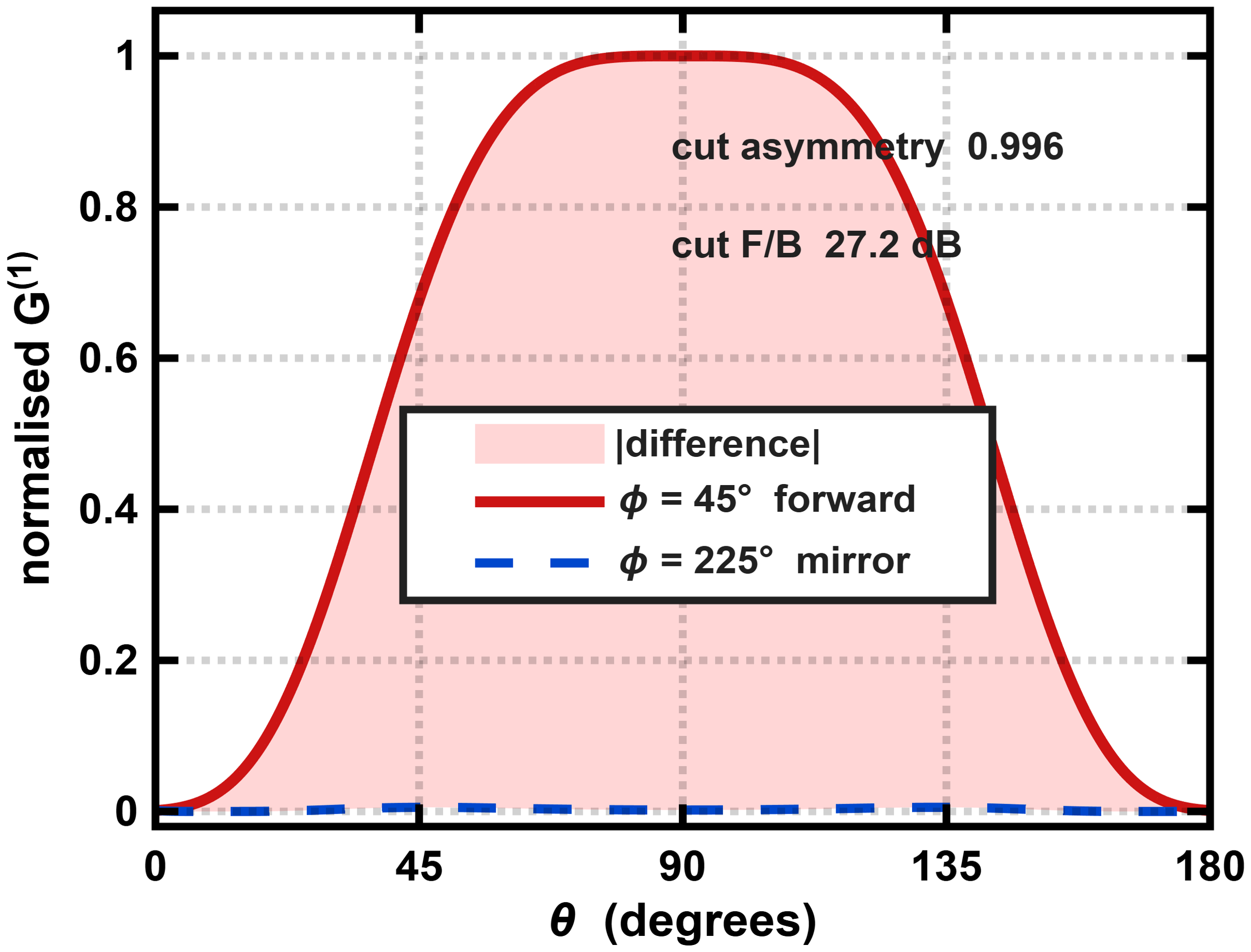}
\caption{Forward and mirrored elevation cuts of the TD-2D state at
$\phi=45^\circ$ and $\phi=225^\circ$. The shaded area is the integrand
of~\eqref{eq:eta2d}. The mirror direction is a deep null, giving a $27.2$~dB
front-to-back ratio and $\eNR^{2\mathrm{D}}=1.867$, i.e.\ $93.3\%$ of the
bound of Theorem~\ref{thm:bound}.}
\label{fig:cut2d}
\end{figure}

\begin{figure}[!t]\centering
\includegraphics[width=\columnwidth]{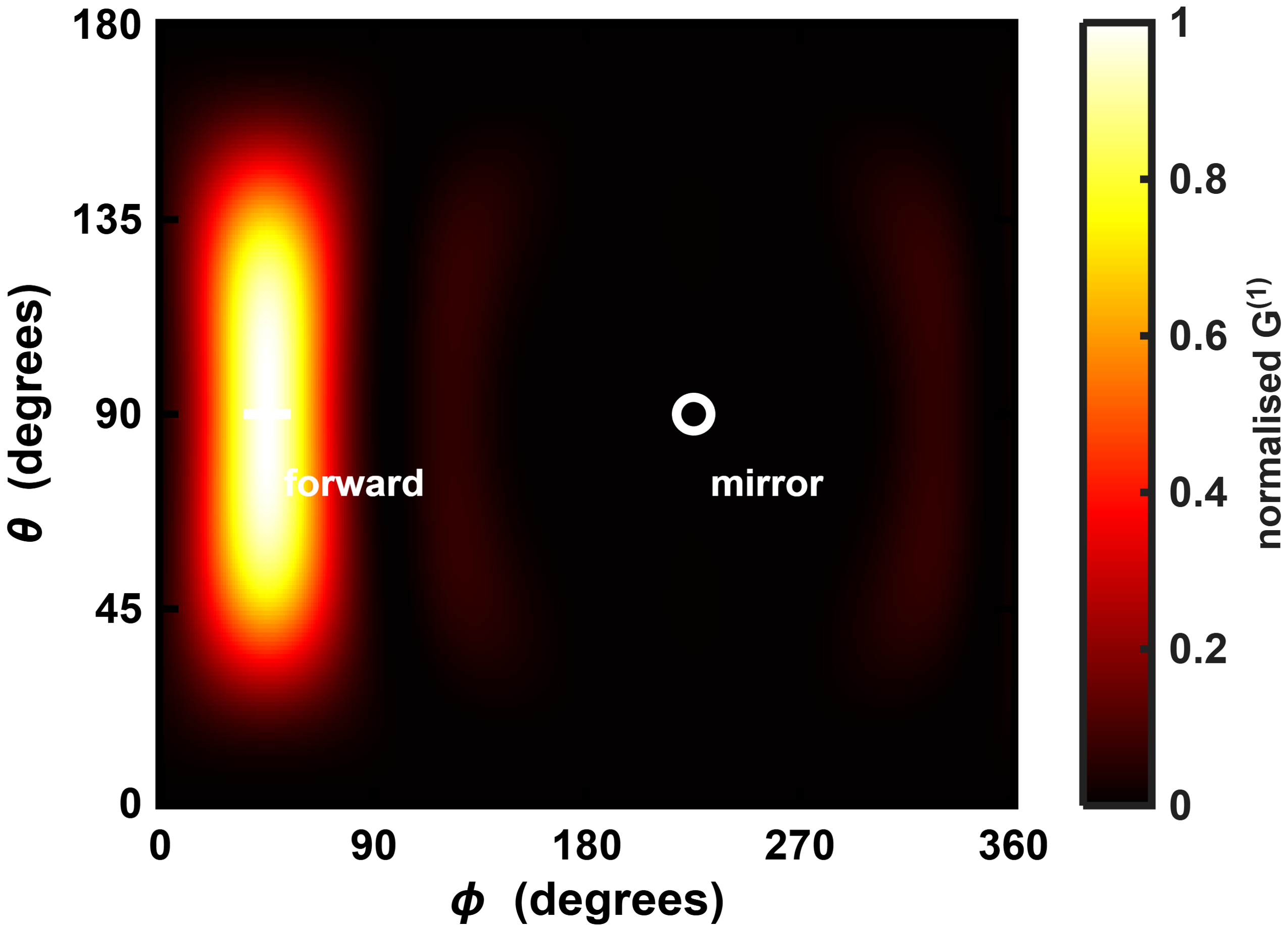}
\caption{Full-sphere $\Gone(\theta,\phi)$ for the TD-2D state. The beam peaks at
$(90^\circ,45^\circ)$; the spatial-reversal mirror at $(90^\circ,225^\circ)$ is a deep
null, giving the $27.2$~dB front-to-back ratio that underpins
$\eNR^{2\mathrm{D}}=1.867$.}
\label{fig:map2d}
\end{figure}

\subsection{Collective Decay Rates}
Fig.~\ref{fig:super} plots $\lambda_{\max}(\Gamma_{ij})$ against spacing for both
arrays. The Dicke value $N\Gamma$ is recovered only for $d\lesssim0.01\lambda_0$; at the
operating spacing $d=0.3\lambda_0$ the maximum eigenrates are $1.61\Gamma$ and
$2.75\Gamma$, while the timed Dicke states themselves decay at $1.01\Gamma$ and
$2.08\Gamma$. The naive $N^{2}\Gamma$ line is shown for reference and lies one to
two orders of magnitude above the true curve throughout the physically accessible
range. The lower branch shows $\lambda_{\min}$, the subradiant dark state.

\begin{figure}[!t]\centering
\includegraphics[width=\columnwidth]{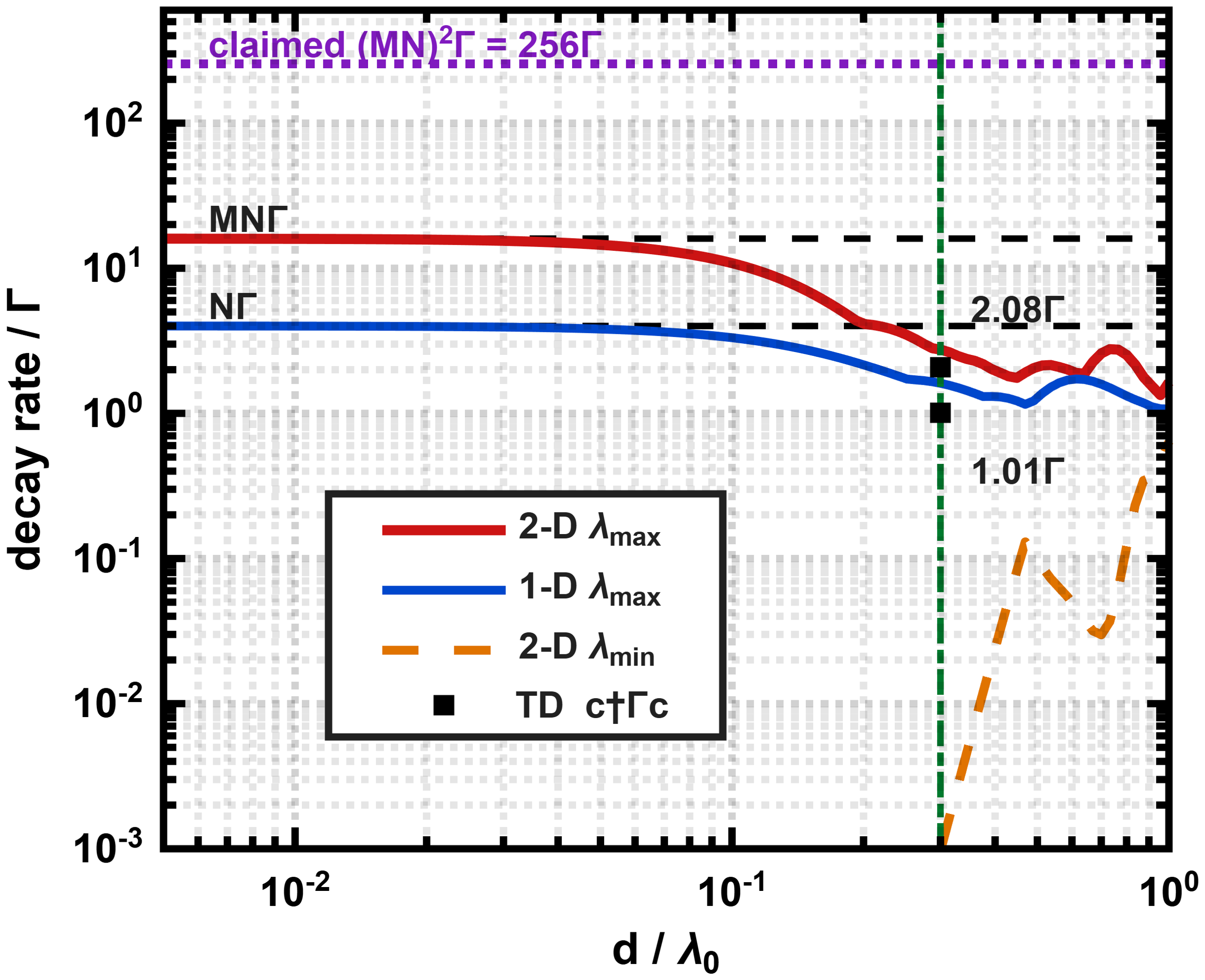}
\caption{Collective decay rates versus spacing, from the eigenvalues
of~\eqref{eq:Gij}. Solid red: $\lambda_{\max}$ for the $4\times4$ array; solid
blue: $\lambda_{\max}$ for the $N=4$ chain; dashed orange: the subradiant
$\lambda_{\min}$ of the $4\times4$ array. The three horizontal lines are, from
the top, the frequently quoted $(MN)^{2}\Gamma=256\Gamma$ and the Dicke limits
$MN\Gamma=16\Gamma$ and $N\Gamma=4\Gamma$. Square markers at the design spacing
$d=0.3\lambda_0$ give the timed Dicke projections
$\mathbf{c}^{\dagger}\bm{\Gamma}\mathbf{c}=2.08\Gamma$ and $1.01\Gamma$, which
are the rates the states actually decay at. The $(MN)^{2}\Gamma$ level is never
approached.}
\label{fig:super}
\end{figure}

\begin{figure}[!t]\centering
\includegraphics[width=0.95\columnwidth]{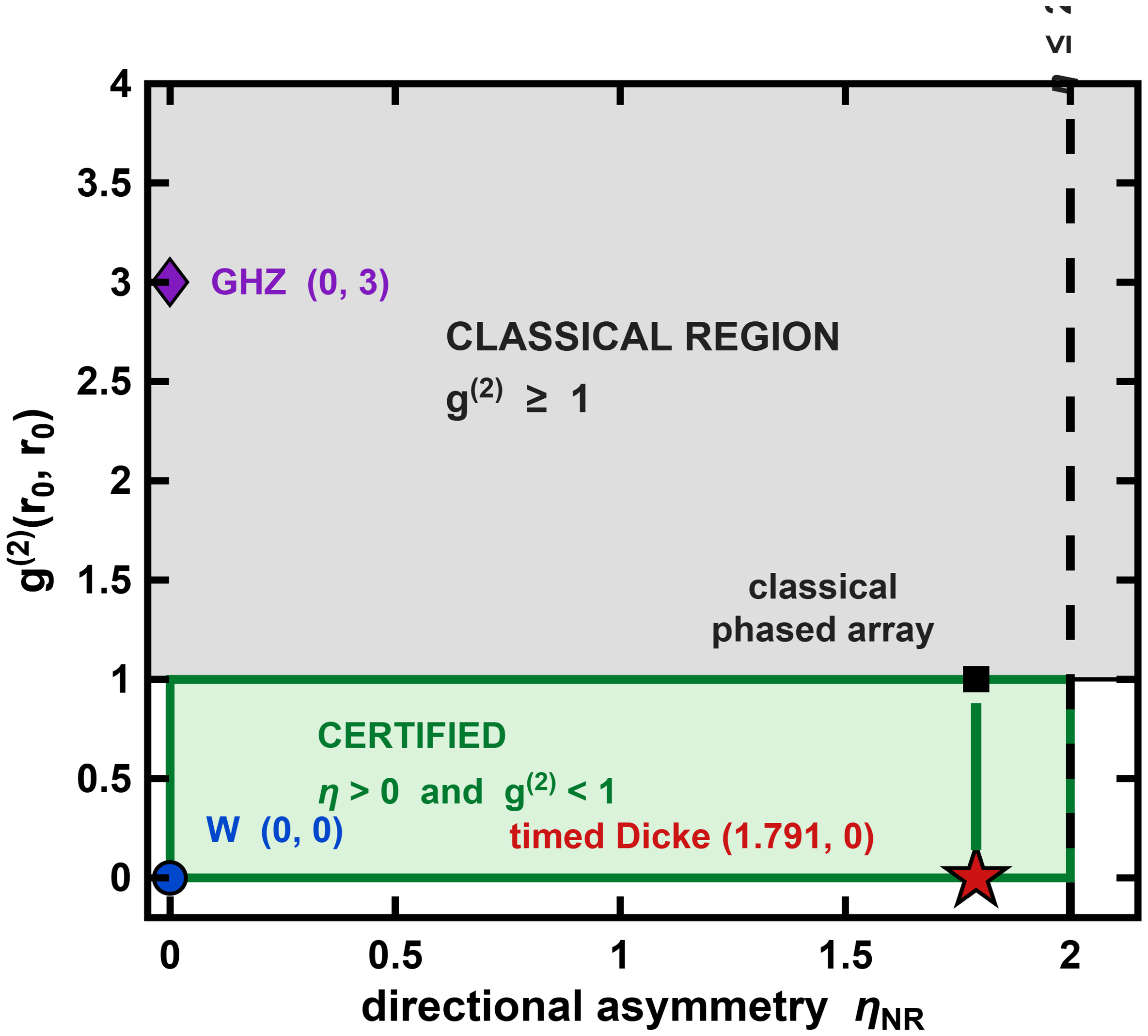}
\caption{The Nonclassical Directional-Emission
Certificate~\eqref{eq:cert} in the $(\eNR,g^{(2)})$ plane. Horizontal position
alone carries no quantum content --- the classical phased array reaches the same
$\eNR$; it is the vertical coordinate that cannot be reached classically. The classical region $g^{(2)}\geq1$ is shaded. A
classical phased array sits at $(1.791,\,\geq\!1)$; the timed Dicke state sits at
$(1.791,\,0)$, inside the certified quantum region. GHZ sits at $(0,\,3)$.}
\label{fig:cert}
\end{figure}

\subsection{Direction-of-Arrival Estimation}
Three levels of realism are reported separately and should not be conflated.
\emph{(a)~Exact-state}: $\mathbf{R}_Q$ is the analytic
model~\eqref{eq:srcmodel}, which bounds what the estimator can do in principle.
\emph{(b)~Finite-shot}: $\widehat{\mathbf{R}}_Q$ is formed from $N_s$ sampled
realisations, which introduces estimation error.
\emph{(c)~Finite-shot with an SNR floor}: a background term
$\sigma_n^2\mathbf{I}$ is added before sampling. Detector dark counts, timing
jitter and imperfect state preparation lie outside all three.

\textbf{Exact-state recovery:}
To first assess the intrinsic performance of the proposed Q-MUSIC
formulation, we consider two sources located at the off-grid directions
$(30.37^\circ,45.62^\circ)$ and $(59.48^\circ,119.31^\circ)$, with
$\sigma_n^2/p_k=10^{-9}$. When the pseudo-spectrum is evaluated on the
discrete angular grid, the estimated directions are
$(30.5^\circ,46.0^\circ)$ and $(59.5^\circ,119.0^\circ)$, resulting in an
RMSE of $0.36^\circ$. This error arises entirely from the finite angular
grid rather than from the Q-MUSIC formulation itself. To separate grid
quantisation from the intrinsic state-recovery capability, we subsequently
replace the discrete search with continuous Nelder--Mead refinement of
(34). The resulting estimates recover both source directions to
$3\times10^{-7}$ degrees. Accordingly, this noiseless result is reported
as numerical-precision recovery under the exact-state model, where the
residual error is determined by the optimiser tolerance rather than by a
physical measurement limit. The deliberately off-grid source locations
ensure that this distinction is not obscured by grid alignment.

\textbf{Finite-shot behaviour:}
The exact-state result above represents an idealised noiseless limit and
therefore does not quantify the statistical error associated with a finite
number of measurements. To evaluate this effect, we construct the
finite-shot estimator $\hat{R}_Q$ from $N_s$ complex Gaussian realisations
with covariance (33) and perform 100 independent trials for each value of
$N_s$. As summarised in Table~\ref{tab:shots}, the angular RMSE decreases approximately
as $N_s^{-1/2}$, consistent with the expected statistical scaling of a
covariance-based estimator.

\begin{table}[ht!]
\renewcommand{\arraystretch}{1.15}
\caption{Q-MUSIC RMSE versus Number of Shots (100 Trials per Point)}
\label{tab:shots}
\centering
\begin{tabular}{cccc}
\toprule
$N_s$ & RMSE$_\theta$ (deg) & RMSE$_\phi$ (deg) & scaling\\
\midrule
$10$   & $1.5\times10^{-3}$ & $1.8\times10^{-3}$ & ---\\
$50$   & $7\times10^{-4}$   & $7\times10^{-4}$   & $0.48$\\
$100$  & $5\times10^{-4}$   & $5\times10^{-4}$   & $0.49$\\
$500$  & $2\times10^{-4}$   & $2\times10^{-4}$   & $0.50$\\
$1000$ & $2\times10^{-4}$   & $2\times10^{-4}$   & $0.50$\\
$5000$ & $1\times10^{-4}$   & $1\times10^{-4}$   & $0.50$\\
\midrule
\multicolumn{4}{@{}p{0.94\columnwidth}@{}}{Continuous peak refinement; the
final column is the fitted exponent of $N_s^{-\alpha}$, confirming
$\alpha\simeq1/2$.}\\
\bottomrule
\end{tabular}
\end{table}

\begin{table}[ht!]
\renewcommand{\arraystretch}{1.15}
\caption{Q-MUSIC RMSE versus SNR ($N_s=100$, 200 Trials per Point)}
\label{tab:snr}
\centering
\begin{tabular}{ccc}
\toprule
SNR (dB) & RMSE$_\theta$ (deg) & RMSE$_\phi$ (deg)\\
\midrule
$0$  & $0.726$ & $0.706$\\
$10$ & $0.214$ & $0.224$\\
$20$ & $0.070$ & $0.067$\\
$30$ & $0.022$ & $0.023$\\
\bottomrule
\end{tabular}
\end{table}

\textbf{Contaminated covariance:}
The finite-shot analysis above assumes that the underlying covariance
matrix $R_Q$ is free from an additional background contribution. To
evaluate the robustness of Q-MUSIC when the measured coherence matrix is
contaminated by such a background floor, we consider the measured
covariance
$R_{\mathrm{meas}}=R_Q+\sigma_n^2 I$. The corresponding signal-to-noise
ratio is defined as
$\mathrm{SNR}=10\log_{10}(p_k/\sigma_n^2)$, while the number of shots is
fixed at $N_s=100$. The resulting RMSE values, obtained over 200 trials
for each SNR, are reported in Table~\ref{tab:snr}. The angular estimation error
decreases approximately by a factor of ten for every 20-dB increase in
SNR and remains below $1^\circ$ even at $0$~dB.
The angular separation of $46.5^\circ$ comfortably exceeds the resolution
limit set by the $13.2^\circ$ broadside beamwidth. The pseudo-spectrum is shown
in Fig.~\ref{fig:qmusic}.

\begin{figure}[!t]\centering
\includegraphics[width=\columnwidth]{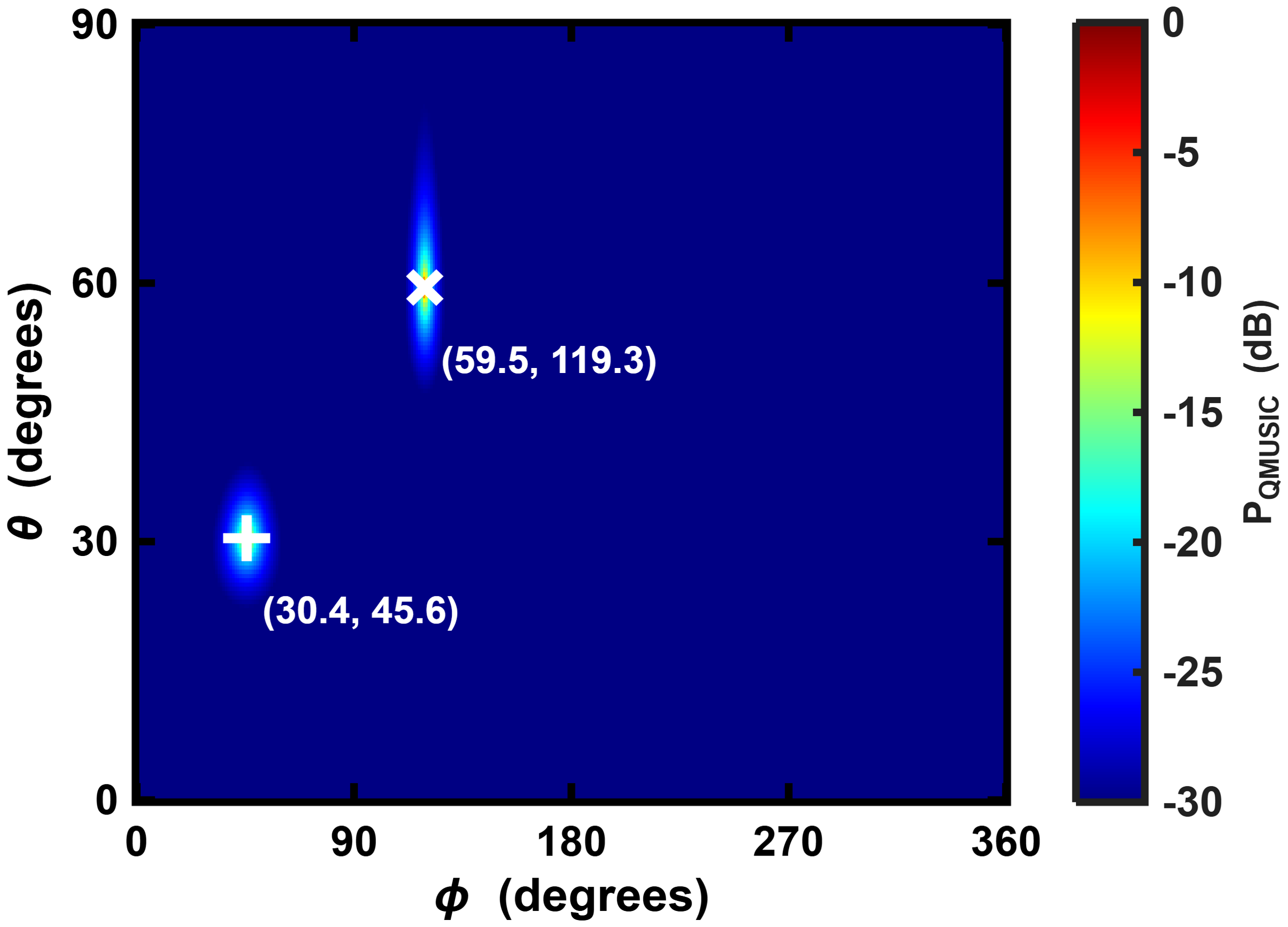}
\caption{Q-MUSIC pseudo-spectrum~\eqref{eq:qmusic} in dB over the upper
hemisphere, computed from the exact-state coherence
matrix~\eqref{eq:srcmodel}. Markers indicate the two off-grid sources at
$(30.37^\circ,45.62^\circ)$ and $(59.48^\circ,119.31^\circ)$; both are
recovered to numerical precision under continuous peak refinement.}
\label{fig:qmusic}
\end{figure}

\subsection{Decoherence}
Pure dephasing at rate $\gamma_\phi$ adds the term
$\gamma_\phi\sum_j(\hat{\sigma}^z_j\rho\hat{\sigma}^z_j-\rho)$
to~\eqref{eq:lindblad}, damping off-diagonal coherences as $e^{-2\gamma_\phi t}$
while leaving populations intact. Since the timed Dicke phase gradient lives
entirely in those coherences, $\eNR^{2\mathrm{D}}$ decays monotonically
(Table~\ref{tab:deph}). The design retains $89\%$ of its asymmetry at $\gamma_\phi/\Gamma=0.1$ and
$74\%$ at $0.25$, but only $22\%$ at $\gamma_\phi/\Gamma=1$. Operation therefore requires
$\gamma_\phi/\Gamma\lesssim0.25$, which for InGaAs implies cryogenic
temperatures~\cite{Lodahl}; the effect is not a numerical artefact but the
physical erasure of the state's phase information. The spacing, by contrast, is
largely uncritical: Fig.~\ref{fig:spacing} shows $\eNR$ varies by under $3\%$
over $0.28$--$0.34\lambda_0$ for both arrays, so fabrication tolerance is not the
dominant concern. In the present model dephasing is the dominant \emph{modelled}
limitation; experimentally, nonradiative recombination, spectral diffusion,
phonon interactions, substrate effects and detector background must also be
accounted for.

\begin{table}[!t]
\renewcommand{\arraystretch}{1.18}
\caption{Dephasing Sensitivity of the TD-2D State}
\label{tab:deph}
\centering
\begin{tabular}{ccc}
\toprule
$\gamma_\phi/\Gamma$ & $\eNR^{2\mathrm{D}}$ & Retention\\
\midrule
$0$    & $1.867$ & $100.0\%$\\
$0.05$ & $1.764$ & $94.5\%$\\
$0.10$ & $1.663$ & $89.1\%$\\
$0.25$ & $1.373$ & $73.5\%$\\
$0.50$ & $0.956$ & $51.2\%$\\
$1.00$ & $0.411$ & $22.0\%$\\
$2.00$ & $0.061$ & $3.3\%$\\
$5.00$ & $0.000$ & $0.0\%$\\
\bottomrule
\end{tabular}
\end{table}

\begin{figure}[!t]\centering
\includegraphics[width=\columnwidth]{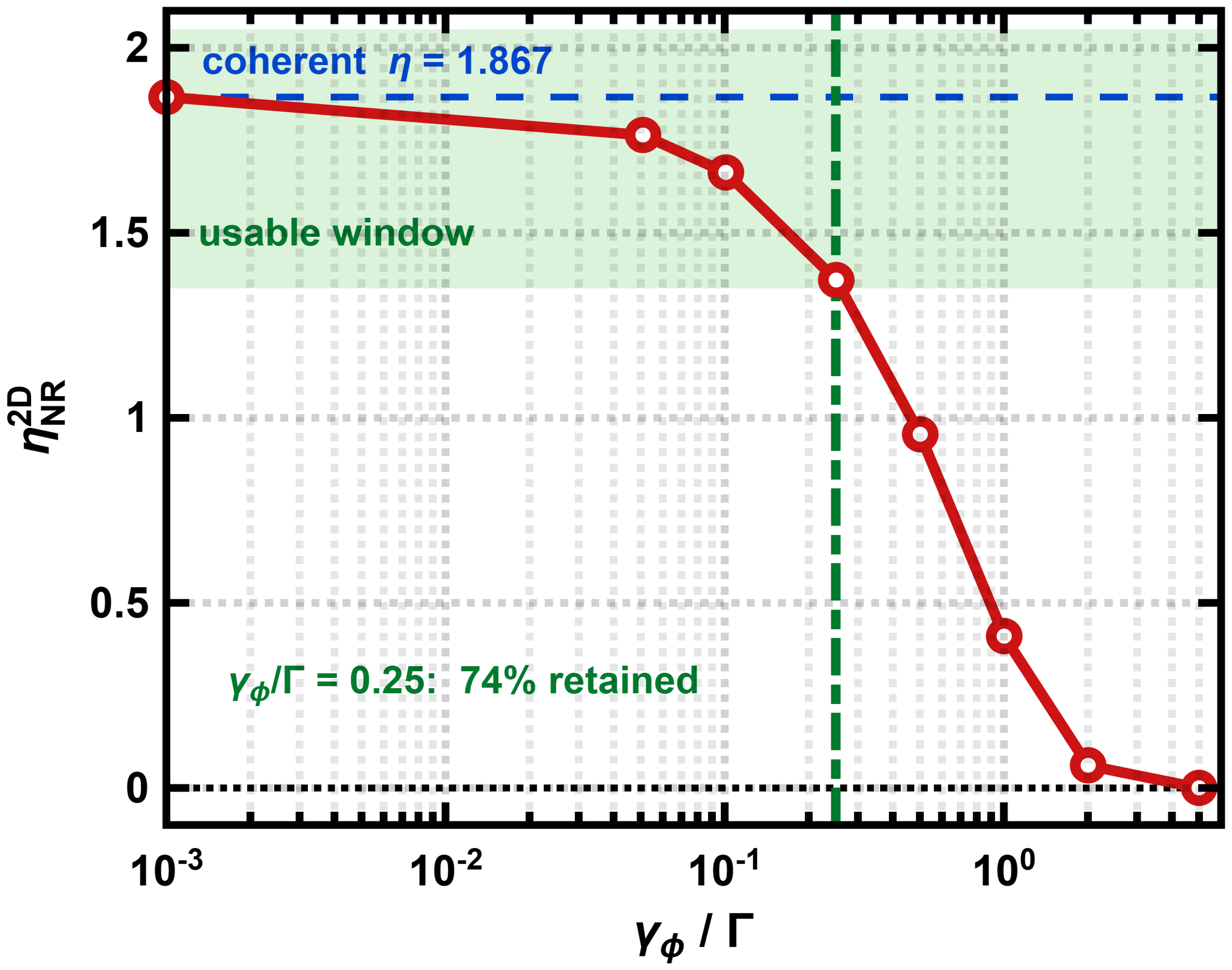}
\caption{Directional asymmetry versus pure dephasing. The dashed line marks the
coherent value $\eNR^{2\mathrm{D}}=1.867$; the dotted line the classical limit.
The usable window is $\gamma_\phi/\Gamma\lesssim0.25$.}
\label{fig:deph}
\end{figure}

\begin{figure}[!t]\centering
\includegraphics[width=\columnwidth]{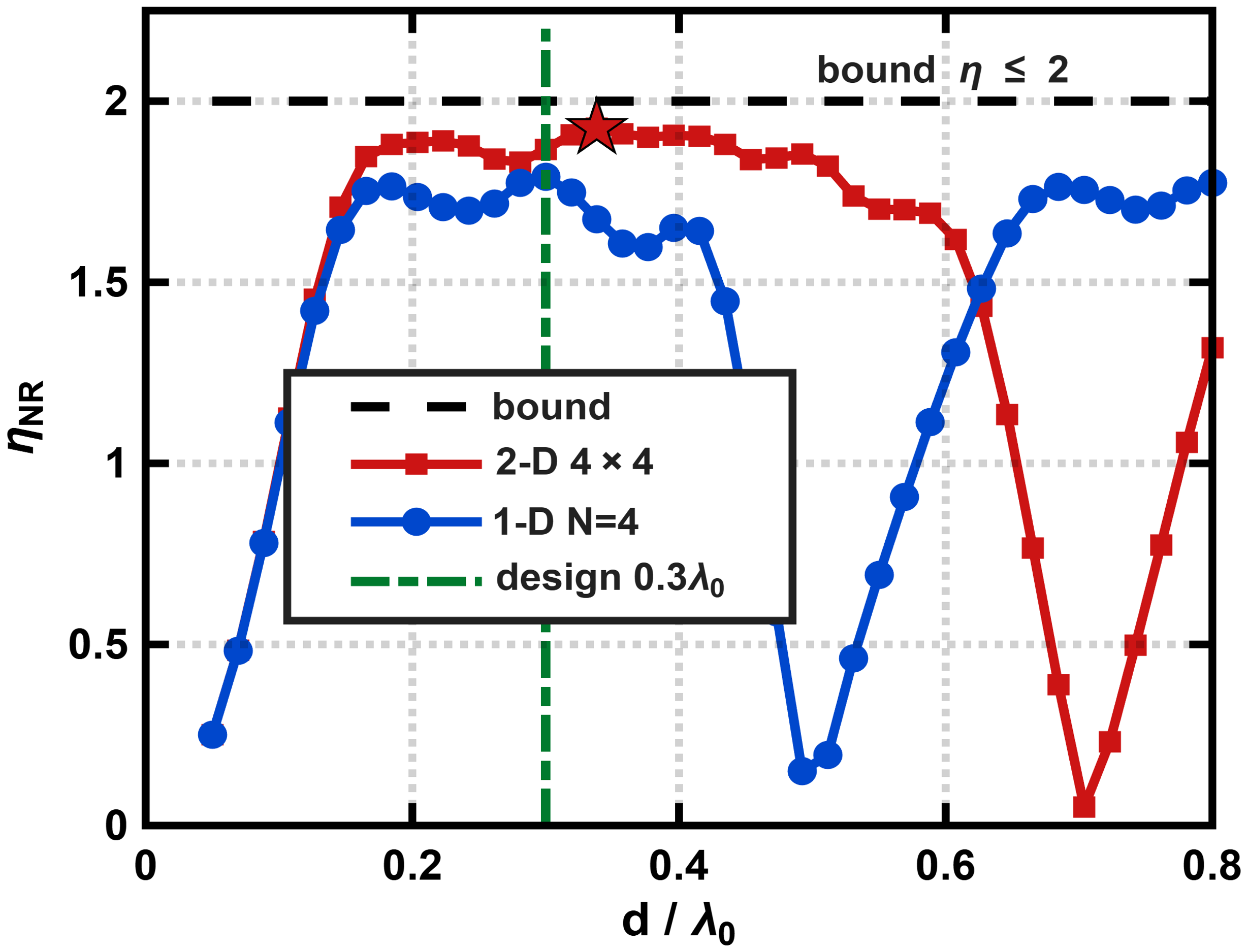}
\caption{Directional asymmetry versus inter-element spacing for both arrays, plotted
against the absolute ceiling $\eNR=2$ of Theorem~\ref{thm:bound}. Stars mark the
optima ($0.300\lambda_0$ in 1-D, $0.338\lambda_0$ in 2-D). The 1-D optimum coincides
with the design point; the 2-D optimum would give $\eNR^{2\mathrm{D}}=1.926$,
only $3\%$ above the design value.}
\label{fig:spacing}
\end{figure}

\subsection{Photon Statistics and the Certificate}
Table~\ref{tab:g2} reports $\Gone$ and $g^{(2)}$ computed in the full $2^{N}$
Hilbert space. The timed Dicke and W states give $\Gtwo=0$ identically, so $g^{(2)}=0$
wherever $\Gone>0$: a one-excitation state cannot produce a coincidence. At
exact pattern nulls the normalised $g^{(2)}$ is undefined ($0/0$). Combined with
$\eNR=1.791$, the certificate~\eqref{eq:cert} is satisfied, and no classical
phased array can occupy this point in the $(\eNR,g^{(2)})$ plane
(Fig.~\ref{fig:cert}). The GHZ state gives $\Gone=N/2=2.0$ independent of angle
and $g^{(2)}=3.0$, exceeding even the thermal value of $2$. Only the timed Dicke state passes. The W state is antibunched but symmetric, the
classical array is asymmetric but never antibunched, and GHZ is neither. The
certificate, therefore, separates precisely the intended physics.

\begin{table}[!t]
\renewcommand{\arraystretch}{1.18}
\caption{Photon Statistics, Full $2^{N}$ Hilbert Space, $N=4$}
\label{tab:g2}
\centering
\setlength{\tabcolsep}{4pt}
\begin{tabular}{lcccl}
\toprule
State & $\Gone$ shape & $g^{(2)}(\hat{\mathbf{r}}_0)$
& $\eNR$ & Certificate\\
\midrule
Classical phased & $|AF|^2$ & $\geq1$ & $1.791$ & fails$^{a}$\\
Coherent state   & $|AF|^2$ & $=1$    & $1.791$ & fails\\
Thermal field    & $|AF|^2$ & $=2$    & $1.791$ & fails\\
W (symmetric)    & $|AF|^2/N$ & $0.000$ & $0.000$ & fails$^{b}$\\
Timed Dicke      & directional & $\mathbf{0.000}$ & $\mathbf{1.791}$ & \textbf{passes}\\
GHZ              & isotropic, $N/2$ & $4(N{-}1)/N$ & $0.000$ & fails\\
\midrule
\multicolumn{5}{@{}p{0.94\columnwidth}@{}}{$^{a}$asymmetry without antibunching;
 $^{b}$antibunching without asymmetry. Only their conjunction certifies.}\\
\bottomrule
\end{tabular}
\end{table}

\subsection{Comparison with Related Work}
Table~\ref{tab:comp} positions this work. To the best of our knowledge it is the first to
combine a planar QD geometry, a full-sphere asymmetry metric, and a QFI
bound within one framework, and to report the saturation law~\eqref{eq:ratio}.

\begin{table}[!t]
\renewcommand{\arraystretch}{1.18}
\caption{Comparison with Related Quantum Array Studies}
\label{tab:comp}
\centering
\footnotesize
\begin{tabular}{lcccccc}
\toprule
\textbf{Work} & \textbf{2-D} & \textbf{$d$--$d$} & \textbf{$g^{(2)}$} &
\textbf{Bound} & \textbf{DoA} & \textbf{QCRB}\\
\midrule
Slepyan--Boag \cite{SlepyanBoag} & --- & \checkmark & --- & --- & --- & ---\\
Mikhalychev \cite{Mikhalychev}   & --- & \checkmark & \checkmark & --- & --- & ---\\
Masson \cite{Masson}             & \checkmark & \checkmark & --- & --- & --- & ---\\
Kim \textit{et al.} \cite{Kim}   & --- & --- & --- & --- & \checkmark & \checkmark\\
Xu \textit{et al.} \cite{Xu}     & \checkmark & --- & --- & --- & \checkmark & ---\\
Banerjee \cite{Banerjee}         & \checkmark & --- & --- & --- & \checkmark & \checkmark\\
\textbf{This work}
 & \checkmark & \checkmark & \checkmark & \checkmark & \checkmark & \checkmark\\
\bottomrule
\midrule
\multicolumn{7}{@{}p{0.94\columnwidth}@{}}{\footnotesize
2-D: planar geometry; $d$--$d$: dipole--dipole coupling with collective Lamb
shift; $g^{(2)}$: photon-statistics analysis; Bound: an absolute limit on the
asymmetry metric; DoA: direction estimation; QCRB: Fisher-information bound.}\\
\bottomrule
\end{tabular}
\end{table}

\section{Robustness, Discussion and Limitations}\label{sec:VIII-B}
\label{sec:limits}

\subsection{Numerical Convergence and Uncertainty}
The full-sphere integral~\eqref{eq:eta2d} converges rapidly: on grids
$(N_\theta,N_\phi)$ of $(45,90)$, $(91,180)$, $(181,360)$, $(361,720)$ and
$(721,1440)$ the value is $1.8670$, $1.8670$, $1.8669$, $1.8669$, $1.8669$.
We therefore quote $\eNR^{2\mathrm{D}}=1.867\pm0.001$, the uncertainty being
the grid residual rather than a physical spread.

\subsection{Vectorial Element Pattern}
\label{sec:vector}
All results reported up to this point use the isotropic scalar element model as
a \emph{baseline}; this subsection adds the vectorial dipole pattern as an
\emph{extension} and reports both, so the reader can see which conclusions
depend on the element model and which do not. We stress at the outset that the
vectorial element factor is \emph{not} a full electromagnetic model: it
corrects the free-space element pattern only, and does not account for a
semiconductor substrate, near-field mutual coupling beyond the dyadic terms
already in~\eqref{eq:Gij}--\eqref{eq:Jij}, or any dielectric environment. Those
require a substrate Green function and remain future work.\\
The baseline treats each QD as a point emitter with an isotropic scalar
radiation factor. A real electric dipole of unit orientation $\hat{\mathbf{d}}$
has the far-field element factor
\begin{equation}
E(\hat{\mathbf{r}})=1-\bigl|\hat{\mathbf{d}}\cdot\hat{\mathbf{r}}\bigr|^{2},
\label{eq:elem}
\end{equation}
equal to $\sin^2\theta$ for $\hat{\mathbf{d}}=\hat{z}$, so that the full pattern
is $E(\hat{\mathbf{r}})\,\Gone(\hat{\mathbf{r}})$. Because
$|\hat{\mathbf{d}}\cdot\hat{\mathbf{r}}|^{2}$ is invariant under
$\hat{\mathbf{r}}\rightarrow-\hat{\mathbf{r}}$, we have the following.

\begin{lemma}
\label{lem:elem}
The vectorial element factor~\eqref{eq:elem} is inversion symmetric, so it
cannot by itself generate directional asymmetry: for any state,
$\eNR=0$ under the scalar model implies $\eNR=0$ under the vectorial model.
\end{lemma}

The element pattern therefore reweights $\Gone$ but leaves the mirror structure
of Sec.~\ref{sec:eta2d} intact, and both Theorem~\ref{thm:bound} and
certificate~\eqref{eq:cert} carry over unchanged. Only the numerical values
shift. Table~\ref{tab:vector} quantifies this for three dipole orientations. For dipoles normal to the array plane --- the natural orientation for a
self-assembled InGaAs layer --- the asymmetry \emph{increases} slightly to
$1.885$ and the directivity gains $1.14$~dB, because $\sin^2\theta$ suppresses
the polar directions where the timed Dicke pattern carries little power.

\begin{table}[ht!]
\renewcommand{\arraystretch}{1.15}
\caption{Effect of the Vectorial Element Pattern, $4\times4$ TD-2D}
\label{tab:vector}
\centering
\setlength{\tabcolsep}{4pt}
\begin{tabular}{lcccc}
\toprule
Element model & $\eNR^{2\mathrm{D}}$ & $D_0$ (dBi) & peak $(\theta,\phi)$ & F/B (dB)\\
\midrule
isotropic scalar          & $1.867$ & $9.49$  & $(90,45)$  & $27.2$\\
$\hat{z}$ (out of plane)  & $1.885$ & $10.63$ & $(90,45)$  & $27.2$\\
$\hat{x}$ (in plane)      & $1.856$ & $9.55$  & $(117,53)$ & $24.4$\\
in-plane diagonal         & $1.659$ & $9.21$  & $(136,45)$ & $20.9$\\
\midrule
\multicolumn{5}{@{}p{0.94\columnwidth}@{}}{W-state control: $\eNR=1.0\times10^{-15}$
for $\hat{\mathbf{d}}=\hat{z}$, confirming Lemma~\ref{lem:elem} numerically.}\\
\bottomrule
\end{tabular}
\end{table}

In-plane dipoles are mildly worse and tilt the beam, since the element null then
lies closer to the intended beam direction; the worst case examined, an in-plane
diagonal dipole whose null is nearest the beam, still gives $\eNR=1.659$,
i.e.\ $83\%$ of the bound. The central results are therefore robust to the
element model, and are not artefacts of the scalar approximation.
Fig.~\ref{fig:vector} shows the corresponding elevation cuts.

\begin{figure}[!t]\centering
\includegraphics[width=\columnwidth]{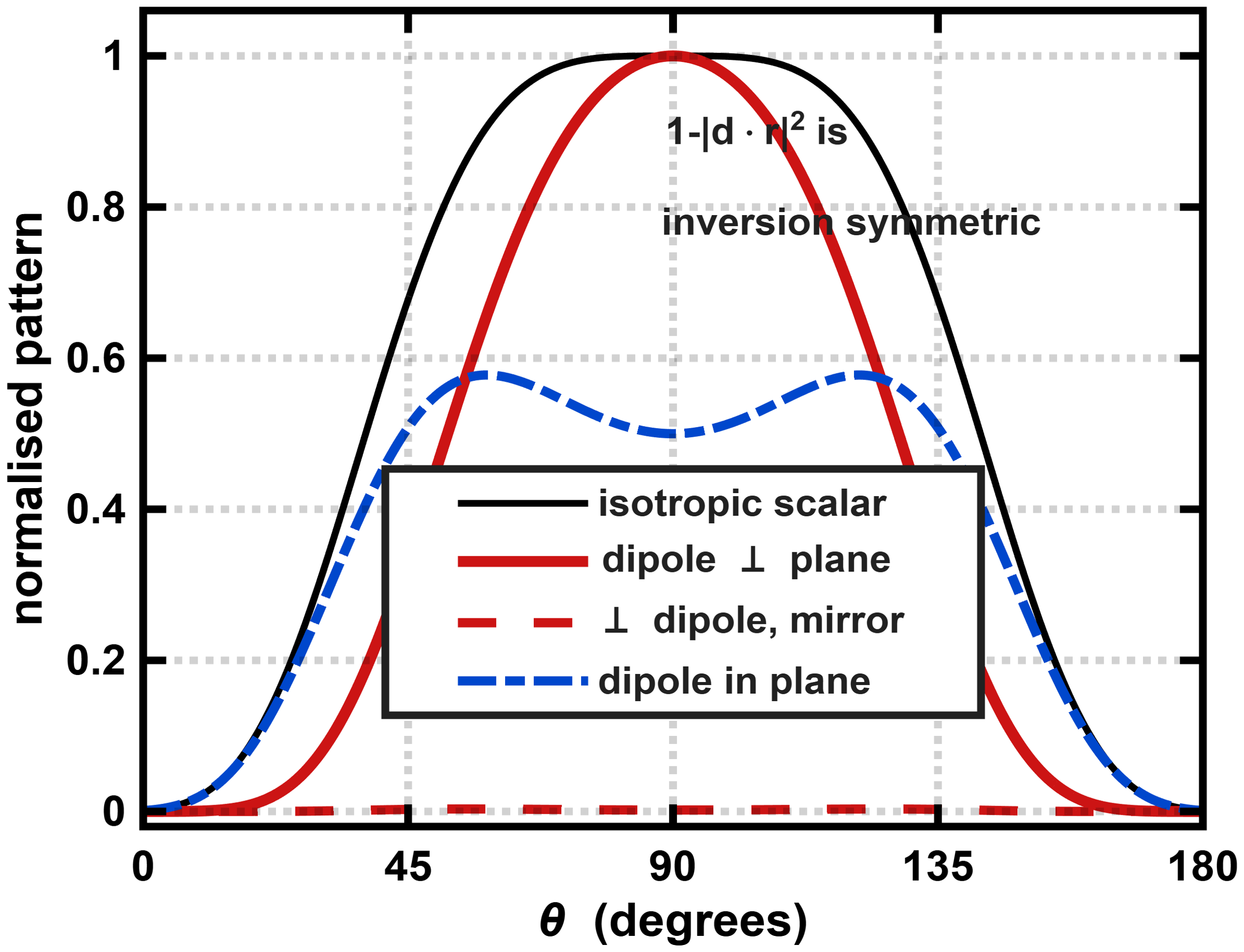}
\caption{Elevation cuts at $\phi=45^\circ$ for the TD-2D state under the scalar
element model and two dipole orientations. The element factor~\eqref{eq:elem}
reweights the pattern but, being inversion symmetric
(Lemma~\ref{lem:elem}), cannot create asymmetry; the mirror curve is shown
dashed for the $\hat{z}$ case.}
\label{fig:vector}
\end{figure}

\subsection{Robustness of the Null}
The front-to-back ratio is a \emph{point} quantity and must be read as such.
Sweeping the observation direction about the exact mirror by
$\Delta\phi=0.1^\circ,0.5^\circ,1^\circ,2^\circ,5^\circ$ changes the mirrored
intensity by less than $12\%$ and the ratio stays within $0.5$~dB of
$27.2$~dB, so the null is angularly robust at the resolution of any realistic
measurement. It is, however, sensitive to the array being exactly periodic;
see below.

\subsection{Phase, Position and Spectral Disorder}
Because the timed Dicke state is defined by phase relationships, phase noise is
the dominant fabrication concern. We add independent Gaussian phase errors
$c_j\rightarrow c_je^{i\delta\phi_j}$ of standard deviation $\sigma_\phi$, and
independent Gaussian positional displacements of standard deviation
$\sigma_r$, and report the mean and standard deviation of $\eNR^{2\mathrm{D}}$
over $N_{\mathrm{MC}}=200$ realisations at each level
(Table~\ref{tab:disorder}).

\begin{table}[ht!]
\renewcommand{\arraystretch}{1.15}
\caption{Disorder Sensitivity of $\eNR^{2\mathrm{D}}$
($N_{\mathrm{MC}}=200$, mean $\pm$ std)}
\label{tab:disorder}
\centering
\begin{tabular}{ccc|ccc}
\toprule
$\sigma_\phi$ & $\eNR^{2\mathrm{D}}$ & ret. &
$\sigma_r/\lambda_0$ & $\eNR^{2\mathrm{D}}$ & ret.\\
(deg) & & (\%) & & & (\%)\\
\midrule
$0$  & $1.867\pm0.000$ & $100$ & $0.00$ & $1.867\pm0.000$ & $100$\\
$5$  & $1.864\pm0.018$ & $99.8$& $0.01$ & $1.860\pm0.018$ & $99.6$\\
$10$ & $1.846\pm0.036$ & $98.9$& $0.02$ & $1.838\pm0.039$ & $98.4$\\
$20$ & $1.787\pm0.070$ & $95.7$& $0.05$ & $1.719\pm0.097$ & $92.1$\\
$30$ & $1.693\pm0.106$ & $90.7$& $0.10$ & $1.517\pm0.152$ & $81.3$\\
$45$ & $1.452\pm0.198$ & $77.8$& ---    & ---             & ---\\
\bottomrule
\end{tabular}
\end{table}

The two mechanisms are not equivalent, and the difference is informative.
A phase error $\delta\phi_j$ enters the amplitude directly and is unbounded in
its effect; a positional error $\delta\mathbf{r}_j$ enters only through
$k_0\hat{\mathbf{r}}\cdot\delta\mathbf{r}_j$, i.e.\ it is a phase error
\emph{projected} onto the observation direction and therefore weighted by
$|\cos\angle(\hat{\mathbf{r}},\delta\mathbf{r}_j)|\leq1$. A displacement of
$\sigma_r=0.05\lambda_0$ thus produces an effective phase spread of at most
$2\pi(0.05)=18^\circ$, and the measured retentions agree: $92\%$ at
$\sigma_r=0.05\lambda_0$ against $96\%$ at $\sigma_\phi=20^\circ$. The practical
reading is that $\eNR$ is limited by the fidelity of the \emph{phase profile},
not by lattice regularity, so fabrication tolerance on element placement is the
less demanding of the two requirements. Dephasing, by contrast, degrades the
asymmetry far more sharply (Fig.~\ref{fig:deph}), and is the binding
constraint. The spread grows
with the disorder level, so the mean alone understates the risk at large
$\sigma$. Table~\ref{tab:limits}
summarises what the present model does and does not include.

\begin{table}[!t]
\renewcommand{\arraystretch}{1.15}
\caption{Scope of the Present Model}
\label{tab:limits}
\centering
\footnotesize
\setlength{\tabcolsep}{4pt}
\begin{tabular}{@{}p{1.55cm}p{2.55cm}p{3.05cm}@{}}
\toprule
\textbf{Aspect} & \textbf{Present model} & \textbf{Future work}\\
\midrule
QD confinement   & 1-D infinite well & 3-D excitonic QD\\
Element pattern  & scalar baseline $+$ vectorial dipole & full dyadic Green function\\
Environment      & homogeneous background & substrate Green function\\
Positional order & Gaussian disorder & correlated disorder\\
Spectral spread  & identical $\omega_0$ & $\omega_j=\omega_0+\Delta_j$\\
Dephasing        & included, Markovian & full phonon model\\
Nonradiative loss& absent & finite $\Gamma_{\mathrm{nr}}$\\
Q-MUSIC input    & exact $\rho$ and finite-shot $\widehat{\mathbf{R}}_Q$ & detector-level model\\
Detection        & ideal, with SNR floor & dark counts, timing jitter\\
QFI              & scalar $\theta$; $\FQM$ evaluated & jointly attainable measurement\\
\bottomrule
\end{tabular}
\end{table}

\subsection{Discussion}
The present work initially focuses on what the $\eNR$ literature actually establishes. The proposed Lemma shows that
pattern asymmetry alone is classically reproducible. Our results suggest that the nonclassical aspect emphasized by Slepyan and
Boag~\cite{SlepyanBoag} is more appropriately associated with the single-photon
character of the excitation than with the first-order spatial pattern itself.
Certificate~\eqref{eq:cert} formalizes that association and makes it experimentally
testable with a Hanbury~Brown--Twiss setup; Sec.~\ref{sec:cert} gives the
proposed protocol. It also exemplifies the advantage of bounded precision. Equation~\eqref{eq:ratio} saturates at $6$.
Heisenberg $1/N$ DoA scaling should not be expected from passive apertures, whose Fisher information is fixed by aperture variance. Additionally, it demonstrates the directivity-asymmetry trade in a 2-D quantum dot array. The asymmetric beam is necessarily end-fire, costing $4.02$~dB of directivity relative to classical
broadside. No such penalty arises in 1-D. Insights on the modal scope can be obtained from \eqref{eq:Gij}--\eqref{eq:Jij}, which assume identical,
parallel, perpendicular dipoles and a homogeneous background. A semiconductor
substrate modifies the Green's function and hence both $\Gamma_{ij}$ and
$J_{ij}$; inhomogeneous broadening across a real QD ensemble would further
degrade $\eNR$. Both of these are beyond the scope of the present work and are left for future extension. The present work also identifies few numerical discrepancies in the prior literature.
For example, evaluating
$d_{ab}=16eL/9\pi^2$ at $L=15$~nm gives $4.33\times10^{-28}$~C$\cdot$m and
$\Gamma=5.46\times10^{6}$~s$^{-1}$. Tabulated values of $4.3\times10^{-29}$ and
$7.3\times10^{5}$ are mutually inconsistent, since the smaller dipole implies
$\Gamma\approx5.3\times10^{4}$. The first sidelobe of a uniform four-element ULA
is $-11.30$~dB, not $-7.43$~dB.


\section{Conclusion}
We have rigorously re-examined directional emission in QD antenna arrays and found the standard figure of merit insufficient. A classical equivalence lemma shows the timed Dicke $G^{(1)}$ coincides with a classical phased-array pattern to $4\times10^{-16}$, so $\eta_{\mathrm{NR}} = 1.791$ is not by itself a quantum signature; nor is pattern asymmetry a test of Lorentz reciprocity, which constrains the transmit--receive response instead. We proved the tight bound $\eta_{\mathrm{NR}} \in [0,2]$, attained to $89.5\%$ in 1-D and $93.3\%$ by a $4\times4$ planar array under the corrected mirror $\phi \to \phi+\pi$, both evaluated with the same solid-angle measure. The principal new result is a sufficient certificate for nonclassical directional emission: $\eta_{\mathrm{NR}} > 0$ together with $g^{(2)}(\hat{\mathbf{r}}_0,\hat{\mathbf{r}}_0) < 1$. Three reference states separate cleanly --- a classical phased array is asymmetric but never antibunched, the symmetric W state is antibunched but symmetric, and only the timed Dicke state combines both. Including dipole--dipole coupling, the timed Dicke decay rate is the projection $\mathbf{c}^\dagger \boldsymbol{\Gamma} \mathbf{c}$, equal to $1.01\Gamma$ and $2.08\Gamma$ rather than $N^2\Gamma$, with a radiatively subradiant mode at $0.001\Gamma$. Closed forms $G^{(1)}_{\mathrm{GHZ}} = N/2$ and $g^{(2)}_{\mathrm{GHZ}} = 4(N-1)/N$ were derived. For sensing, Quantum MUSIC localizes two sources on the quantum spatial coherence matrix, and the elevation QFI obeys $6(M-1)/(M+1)$ for the corner-pair probe --- which, we also showed, has $F_{\phi\phi} = 0$ and is blind in azimuth.\\
Two robustness results support the central claim. Replacing the scalar element factor by the vectorial dipole pattern (Sec.~\ref{sec:VIII-B}) $1 - |\hat{\mathbf{d}}\cdot\hat{\mathbf{r}}|^2$ leaves both Theorem~1 and the certificate intact, because that factor is inversion symmetric (Lemma~2) and so cannot generate asymmetry by itself; for dipoles normal to the array plane the asymmetry in fact rises to $\eta^{2D}_{\mathrm{NR}} = 1.885$ with $1.14$ dB more directivity, and the worst orientation examined still retains $83\%$ of the bound. Under a contaminated coherence matrix $\mathbf{R}_{\mathrm{meas}} = \mathbf{R}_Q + \sigma_n^2 \mathbf{I}$ with $N_s = 100$ shots, Q-MUSIC localizes both sources with RMSE falling from $0.73^\circ$ at $0$ dB SNR to $0.02^\circ$ at $30$ dB.\\
The corrected picture is more modest but sharper: the quantum state does not confer a pattern unavailable to classical feeds, but it delivers that pattern one photon at a time. The remaining work is numerical experiments and electromagnetic rather than conceptual: a substrate dyadic Green function, spectral and positional disorder beyond the Gaussian model used here, nonradiative loss, detector-level noise, and, ultimately, a Hanbury Brown--Twiss test of the certificate. With respect to epistemic status, every number in this paper is derived from theory and computed twice in independent platform implementations that agree to six significant figures. The certificate is offered as a criterion that an experiment \textit{could} apply, with a protocol given in Sec.~\ref{sec:IV-C}. Physical experimentation is left as a future scope of extension.
\bibliographystyle{IEEEtran}
\bibliography{refQuantum}

\end{document}